%% file: VotingFacilityLocation.tex
\pdfoutput=1
\documentclass[a4paper,11pt,,english]{article}

\usepackage{amsmath, amssymb, amsthm, array, caption, color,fullpage, graphicx, hyperref,mathtools, multirow, subcaption, tikz, epigraph}

\begin{document}

\newcommand{\inote}[1]{{\color{red}$\ll$\textsf{#1 --Iddan}$\gg$\marginpar{\tiny\bf IG}}}
\newcommand{\mnote}[1]{{\color{blue}$\ll$\textsf{#1 --Michal}$\gg$\marginpar{\tiny\bf MF}}}
\newcommand{\anote}[1]{{\color{purple}$\ll$\textsf{#1 --Amos}$\gg$\marginpar{\tiny\bf AF}}}

\newcolumntype{?}{!{\vrule width 1.5pt}}

\newtheorem{theorem}{Theorem}
\newtheorem{corollary}[theorem]{Corollary}
\newtheorem{lemma}[theorem]{Lemma}
\newtheorem{observation}[theorem]{Observation}
\newtheorem{conjecture}[theorem]{Conjecture}
\newtheorem{proposition}[theorem]{Proposition}
\newtheorem{definition}[theorem]{Definition}
\newtheorem{claim}[theorem]{Claim}

\newcommand{\cost}{\mbox{\rm cost}}
\newcommand{\SC}{\mbox{\rm SC}}
\newcommand{\OPT}{\mbox{\rm OPT}}
\newcommand{\bara}{\mathbf{a}}
\newcommand{\bx}{\mathbf{x}}
\newcommand{\hx}{\hat{x}}
\newcommand{\cala}{\mathcal{A}}
\newcommand{\calc}{\mathcal{C}}
\newcommand{\cals}{\mathcal{S}}
\newcommand{\calv}{\mathcal{V}}
\newcommand{\calr}{\mathcal{R}}
\newcommand{\bij}{b_{i,j}}
\newcommand{\xmed}{x_{\ceil*{n/2}}}
\newcommand{\hbx}{\hat{\bx}}
	
\usetikzlibrary{plotmarks,automata,positioning}
\usetikzlibrary{decorations.markings}
\tikzstyle{axis}   = [-latex,black!55]
\tikzstyle{two}=[x={(1cm,0cm)},y={(0cm,1cm)}]
\tikzstyle{dim}    = [latex-latex]

\DeclarePairedDelimiter{\ceil}{\lceil}{\rceil}
\DeclarePairedDelimiter{\floor}{\lfloor}{\rfloor}


\title{On Voting and Facility Location}
\author{
	Michal Feldman\thanks{Tel Aviv University and Microsoft Research, Israel. Email: {\tt michal.feldman@cs.tau.ac.il}.
		The work of Michal Feldman was partially supported by the European Research Council under the European Union's Seventh Framework Programme (FP7/2007-2013) / ERC grant agreement number 337122.} \and
	Amos Fiat\thanks{Tel Aviv University. Email: {\tt fiat@tau.ac.il}.
		This work was done in part while A. Fiat was visiting the Simons Institute for the Theory of Computing.} \and
	Iddan Golomb\thanks{Tel Aviv University. Email: {\tt igolomb@gmail.com}.
		The work of Iddan Golomb was partially supported by the European Research Council under the European Union's Seventh Framework Programme (FP7/2007-2013) / ERC grant agreement number 337122. This work was done in part while I. Golomb was visiting the Simons Institute for the Theory of Computing.}}

\date{December 18, 2015}
\maketitle

\begin{quote}
{\small\raggedright \noindent {\em ``We would all like to vote for the best man, but he is never a candidate''}}

\hspace{2cm} {\tiny\raggedleft \noindent {--- Kin Hubbard}}
\end{quote}

\input{abstract_sv}

\input{intro_sv}
\input{model_sv}

\input{classes_sv}

\input{spike_sv}

\input{randomized_sv}

\input{deterministic_sv}

\input{discussion_sv}

\section*{Acknowledgements}
We would like to thank Alon Eden, Ilan Cohen, Ophir Friedler, Shai Vardi, Rachel Matichin, Sella Nevo, Guy Reiner, Lirong Xia and Elliot Anshelevich for interesting discussions.

\bibliographystyle{plain}
\bibliography{refs}

\newpage
\appendix
\input{appendix.tex}

\end{document}

%% file: abstract_sv.tex
\begin{abstract}
	 We study mechanisms for candidate selection that seek to minimize the social cost, where voters and candidates are associated with points in some underlying metric space. The social cost of a candidate is the sum of its distances to each voter. 
	 Some of our work assumes that these points can be modeled on a real line, but other results of ours are more general.

	A question closely related to candidate selection is that of minimizing the sum of distances for facility location.
	The difference is that in our setting there is a fixed set of candidates, whereas the large body of work on facility location seems to consider every point in the metric space to be a possible candidate.
	This gives rise to three types of mechanisms which differ in the granularity of their input space (voting, ranking and location mechanisms).
	We study the relationships between these three classes of mechanisms.
	
	While it may seem that Black's 1948 median algorithm is optimal for candidate selection on the line, this is not the case.
	We give matching upper and lower bounds for a variety of settings.
	In particular, when candidates and voters are on the line, our universally truthful {\sl spike} mechanism gives a [tight] approximation of two.
	When assessing candidate selection mechanisms, we seek several desirable properties: $(a)$ efficiency (minimizing the social cost) $(b)$ truthfulness (dominant strategy incentive compatibility) and $(c)$ simplicity (a smaller input space).
	We quantify the effect that truthfulness and simplicity impose on the efficiency.

\end{abstract}

%% file: intro_sv.tex
\section{Introduction}

The Hotelling-Downs model (\cite{downs1957economic}, \cite{hotelling1990stability}) used to study political strategies, assumes that individual voters occupy some point along the real line. Non-principled political parties (or ice cream vendors) strategically position themselves at a point 
along the left-right axis (or along a beach) so as to garner the greatest number of supporters (clients). Implicitly, voters will vote for the closest candidate.

We consider an analogous problem to the Hotelling-Downs model, where candidates are principled ({\sl i.e.}, non-strategic) whereas the voters have preferences but may misrepresent them in order to achieve what is a better outcome from their perspective. In this model, in which both voters and candidates are represented by points in the metric space, a closer candidate is preferable to one further away.

Examples for candidate selection:
\begin{itemize}
	\item A municipality plans to erect a public library on a street, and every resident seeks to be as close as possible to the proposed library. However, the new library can only be built on suitable locations (the candidates).
	\item Social choice issues in which the distance is not physical:  there is a set of policies ranging from left to right, and several political candidates stand for election, each one advocating a different policy.
	Every voter is associated with a point along the real line.
	An example of a collective decision problem which does not revolve around the political sphere yet may also fit this setting is the task of determining the temperature of an air conditioner in a room, where each individual has a different ideal point along the scale of temperatures (a line). 
	There are many additional settings of relevant candidate selection problems, e.g., in the realms of recommendation systems, electronic commerce and computational economics.
	While our results do not necessarily apply to all social choice settings, there are many such problems for which they do apply (whether in entirety or partially).
\end{itemize}

Assuming quasi-linear utilities, and allowing payments --- then the well known Vickrey-Clarke-Groves (VCG) mechanism is truthful and can achieve the optimal social cost (see, e.g., \cite{nisan2007introduction}).
However, in many real-life situations we restrict the use of money due to ethical, legal or other considerations, e.g, in democratic elections and in the examples previously mentioned.

We study deterministic truthful mechanisms with no payments for the candidate selection problem. In such mechanisms, no agent can benefit from misreporting her location, regardless of the reported locations of the other agents. Such mechanisms are also known as dominant strategy incentive compatible mechanisms. We also consider randomized truthful mechanisms, both universally truthful (ex-post Nash) and truthful in expectation.

Given a set of candidate and voter locations, it is polytime to find the candidate that minimizes the social cost.

When restricted to deterministic truthful mechanisms, we show that the optimal candidate cannot be selected in the general case. Moreover, we show that the cost may be as bad as three times the optimal cost (matching lower and upper bounds).
When considering randomized mechanisms on the line, the approximation factor drops to two (matching upper and lower bounds).

There are other reasons that an optimal candidate may not be chosen. In particular, this depends on the amount of information the agents supply to the mechanism. We formulate three different types of mechanisms, based on the information each agent submits to the mechanism---
\begin{itemize}
	\item Voting mechanisms, in which each agent casts a vote for her favorite candidate.
	\item Ranking mechanisms, in which each agent states her ordinal preferences over the candidates.	
	\item Location mechanisms, in which each agent sends her exact position.
\end{itemize}

Clearly, knowing the true location of an agent allows one to infer the ranking preferred by that agent, which in turn allows one to infer the favorite candidate of the agent (up to tie-breaking).

In almost all previous work on the facility location problem every point in the metric space was considered to be a candidate, therefore there was no difference between these three mechanism types.

The social choice literature mostly considers ranking mechanisms. Recognize that Arrow's impossibility theorem does not hold when assuming the preferences are single-peaked. 

The more information an agent transmits, the mechanism has more tools to devise an accurate solution.
Albeit, this information comes at a cost, since it might disclose more private information which the agents wish to keep confidential.
Furthermore, behavioral economists have long argued that the agents cannot obtain the full information pertaining to their utility, or that obtaining this information requires a high cognitive cost.
Additionally, sending more information also casts a higher burden on the mechanism.
For all of these reasons deploying a simple mechanism \footnote[1]{We use the term ``simplicity" in the perspective of the voters, which have a smaller action space, i.e, less options to choose from. Upon receiving the input, the mechanism itself can act in an arbitrarily complex fashion.} which requires less information from agents is advantageous, and generally there is a trade-off between the accuracy of a mechanism to its simplicity. 
Indeed, in practice many election schemes use voting mechanisms rather than ranking mechanisms, largely due to these desiderata.

\subsection{Our Contributions}
In the paper, we show the following:

\begin{itemize}
	\item In Section \ref{sec-classes} we formulate a framework of reductions that compare the various mechanism types. 
	We utilize this framework to show the relations (equivalence or strict containment) between the three classes of truthful mechanisms -- voting, ranking and location (see Figure \ref{fig:classes}).
	Furthermore, we show that for the case of two candidates, the set of truthful in expectation location mechanisms is equivalent to the set of truthful in expectation voting mechanisms.
	These results provide a significant step towards a full characterization of truthful mechanisms at large.

	\begin{figure}[t]
		
		\setlength{\fboxrule}{0.5pt}
		\noindent \fbox{\noindent\makebox[1\textwidth][c]
			{\begin{minipage}{1\textwidth}

					\centering	
					\begin{subfigure}{.5\textwidth}
						\centering
						
						\begin{tikzpicture}
						\draw[fill=black] (-0.7, -0.4) circle (0.05) node  (1) [right] {Median};		
						\draw (0,0) ellipse (1.6cm and 0.7cm) node [above, yshift=-0.1cm] {Voting};
						\draw (0,0) ellipse (2.8cm and 1.9cm) node [above, yshift=0.75cm] {Ranking $\approx$ Location};
						\end{tikzpicture}
						
						\caption{\textbf{Deterministic truthful} mechanisms}
						\label{fig:classes-det}
					\end{subfigure}%
					\begin{subfigure}{.5\textwidth}
						\centering
						
						\begin{tikzpicture}
						\draw[fill=black] (9.7, -0.38) circle (0.05) node  (1) [right] {Spike};
						\draw (10,0) ellipse (1.4cm and 0.7cm) node [above, yshift=-0.1cm] {UT Voting};
						\draw (10,0) ellipse (2.1cm and 1.5cm) node [above, yshift=0.7cm] {TIE Voting};
						\draw (10,0) ellipse (3cm and 2.3cm) node [above, yshift=1.5cm] {TIE Ranking}; 	
						\draw (10,0) ellipse (3.9cm and 3cm) node [above, yshift=2.3cm] {TIE Location};		
						\end{tikzpicture}

						\caption{\textbf{Randomized} mechanisms}
						\label{fig:classes-ran}
						
					\end{subfigure}			

					\caption{The relationships between classes of mechanisms in candidate elections (Theorem \ref{thm:classes}):
						\newline
						For deterministic truthful mechanisms, the set of ranking mechanisms strictly contains the set of voting mechanisms, yet the set of location mechanisms is equivalent to the set of ranking mechanisms. 
						The lower bound on the social cost for any truthful location mechanism on the line is 3, and it is matched by an upper bound by a voting mechanism - the median mechanism.
						\newline
						In the randomized case, there is a hierarchy of strict containment in the following order - truthful in expectation (TIE) location mechanisms, TIE ranking mechanisms, TIE voting mechanisms and universally truthful (UT) voting mechanisms.
						The lower bound on the social cost of any location mechanism on the line is 2. We introduce the spike mechanism, which is a universally truthful voting mechanism, which achieves a matching upper bound of 2. \newline
						Refer to Section \ref{sec-classes} for formal definitions of equivalence and strict containment in our setting.
					}
					
					\label{fig:classes}
					
				\end{minipage}}} \par\setlength{\fboxrule}{0.2 pt}

		\end{figure}
		
		\item In Section $\ref{sec-spike}$ we define a family of universally truthful voting mechanisms on the line called weighted percentile voting (WPV) mechanisms, which choose the $i$'th vote with some predetermined probability $p_i$. We introduce the {\em spike} mechanism, which is a WPV mechanism that carefully crafts the probability distribution to account for misreports by any agent - whether they are near the center or close to the extremes (see Figure \ref{fig:spike-prob}). We then use backwards induction to show that spike achieves an approximation ratio of two (Theorem \ref{thm-spike}).
		
		\begin{figure}[t]		
			\setlength{\fboxrule}{0.5 pt}
			\noindent \fbox{\noindent\makebox[1\textwidth][c]
				{\begin{minipage}{1\textwidth}		
						\centering
						\includegraphics[scale=0.3]{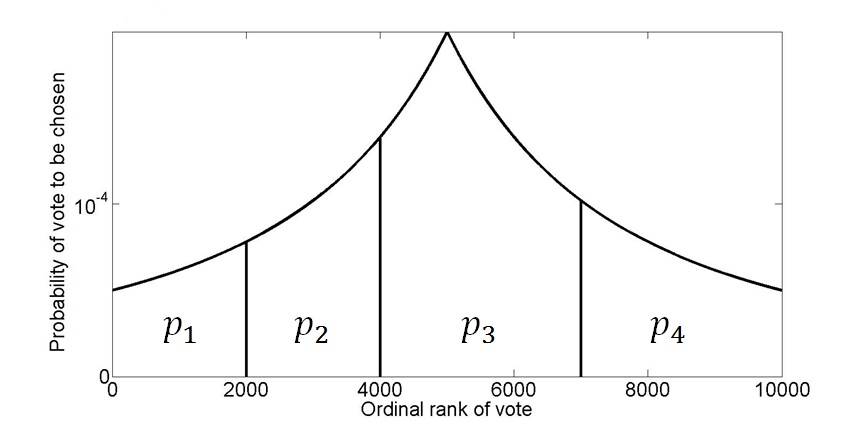}

						\caption{The density function of the spike mechanism, which gives rise to the mechanism's name (the cumulative distribution function is written explicitly in Definition \ref{def:spike}). In this example there are 10000 agents and 4 candidates. The 4 candidates, when ordered from left to right, receive 2000, 2000, 3000 and 3000 votes respectively. The votes are arranged in ascending order (with ties broken arbitrarily), and the graph depicts probability of each vote being chosen -- votes are chosen with higher probability when they are closer to the 50th percentile. The area beneath the graph represents the probability that each candidate will be elected, e.g., the probability of choosing the second candidate ($p_2$) is the integral of the function between 2000 and 4000. }
												
						\label{fig:spike-prob}			
					\end{minipage}}} \par\setlength{\fboxrule}{0.2pt}
		\end{figure}		

		\item In Section \ref{sec-randomized} we show additional bounds for randomized mechanisms -- On the line there is a lower bound of two, even for location mechanisms, which shows that the result for spike is tight. 
		Furthermore, when combining this understanding with the results of Section \ref{sec-classes}, it can be concluded that two is also the tight approximation ratio for truthful in expectation mechanisms (voting, ranking or location) and for universally truthful voting mechanisms.
		
		We move on to show bounds for randomized mechanisms for more general metric spaces\footnotemark[2] (see Figure \ref{fig:results-table-Rd}). An immediate result is that the random dictator mechanism achieves an upper bound of three for any metric space. Theorem \ref{thm:simplex} shows a lower bound of $3-\frac{2}{d}$ for any voting mechanism in $\Re^d$ by using a counterexample based on a regular simplex. This is enough to conclude that on an arbitrary metric space, the bound of three is tight for any voting mechanism. Theorem \ref{thm:triangle-ranking} displays a lower bound of $7/3$ for any ranking mechanism in $\Re^2$ (which also holds in any higher dimension Euclidean space $\Re^d$). 
		
		\footnotetext[2]{We do not present results for deterministic mechanisms in general metric spaces, since in these cases the incentive compatibility constraints take a significant toll on the approximation ratio -- according to Anshelevich et al. \cite{anshelevich2014approximating} in the non-strategic setting it is possible to reach a constant ratio in any metric space, while due to the characterization of Schummer and Vohra \cite{schummer2002strategy} there exist metric spaces in which the approximation ratio is $\Omega(n)$ even in the continuous model.}
		
		\begin{figure}[t]				
			\setlength{\fboxrule}{0.5 pt}
			\noindent \fbox{\noindent\makebox[1\textwidth][c]
			{\begin{minipage}{1\textwidth}	
						
					\centering
										
					\begin{tabular}{|c|c|ll|ll|}
						\cline{3-6} \multicolumn{2}{l|}{ }  & \multicolumn{2}{|c|}{Strategic} & \multicolumn{2}{c|}{Non-Strategic} \\ 
						\hline Voting (low information) & LB
						& $3-\frac{2}{d+1}$ & (Thm. \ref{thm:simplex}) 
						& $2$ & (Lemma \ref{lemma:rand-lb-2-non-strategic}) 
						\\ 
						
						\cline{2-6} & UB 
						& $3$ & (Lemma \ref{lemma:rd-3}) 
						& $3$ & (Lemma \ref{lemma:rd-3}) 
						\\ 
						
						\hline Ranking & LB  
						& $7/3$ & (Thm. \ref{thm:triangle-ranking}) 
						& $2$ & (Lemma \ref{lemma:rand-lb-2-non-strategic}) 
						\\ 
						
						\cline{2-6} & UB  
						& $3$ & (Lemma \ref{lemma:rd-3}) 
						& $3$ & (Lemma \ref{lemma:rd-3}) 
						\\ 
						
						\hline Location (high information) & LB  
						& $2$ & (Obs. \ref{obs:lb-location-2}) 
						& $1$ & 
						\\ 
						
						\cline{2-6} & UB  
						& $3$ & (Lemma \ref{lemma:rd-3}) 
						& $1$ & 
						\\ 
						\hline 
					\end{tabular} 

					\caption{Summary of our results for randomized mechanisms in $\Re^d$.
						 The columns correspond to the truthfulness constraints (in the strategic setting the mechanism must be truthful in expectation), whereas the rows correspond to the information constraints (voting, ranking or location mechanisms). The rows are further divided to show the lower bound (LB) and upper bound (UB) of each such cell. 
						Note that for non-strategic location mechanisms the result is always optimal by definition, since there are neither information nor strategic constraints.
						\newline
						In the strategic setting there is a difference between the lower bounds for voting, ranking and location mechanisms. 
						\newline
						Most of the results here are rather straightforward, except for the upper bounds of $3-\frac{2}{d+1}$ and of $7/3$, which are more involved.
					}					
					
					\label{fig:results-table-Rd}
					
				\end{minipage}}} \par\setlength{\fboxrule}{0.2pt}								
					
		\end{figure}

		\item In Section \ref{sec-deterministic} we present deterministic bounds on the line -- there is a lower bound of three, which is met by a matching upper bound due to the median mechanism. All the results on the line, deterministic or randomized, are displayed in the table in Figure \ref{fig:results-table-line}. 
		
		Recognize the following surprising phenomenon apparent in Figure \ref{fig:results-table-line}. In both deterministic and randomized cases, any constraint in either information or truthfulness, yields the same ratio as taking the both of these constraints simultaneously --- 
		when insisting on truthful mechanisms (in the strategic case), there is no trade-off between high and low information settings, and one can enjoy the benefits of minimal information mechanisms (voting mechanisms) without incurring any additional cost to the approximation ratio;
		Similarly, when deciding to reduce the information requirements to anything less than location mechanisms, it is possible to devise a truthful [voting] mechanism, without increasing the approximation ratio.  
		
		\begin{figure}[t]				
		\setlength{\fboxrule}{0.5 pt}
		\noindent \fbox{\noindent\makebox[1\textwidth][c]
			{\begin{minipage}{1\textwidth}						
			
						\centering	
	
							
						\begin{tabular}{|c|c|c?c|c|}
						\cline{2-5} \multicolumn{1}{l|}{ } & \multicolumn{2}{c?}{ Deterministic }  & \multicolumn{2}{c|}{Randomized} \\ 
						\cline{2-5} \multicolumn{1}{l|}{ }  & Strategic & Non-Strategic & Strategic & Non-Strategic \\ 
						
						\hline 
						Voting & 3 & 3 & 2 & 2 \\
						(low information) & (LB:Lm. \ref{lemma:lb-det-3}, & (LB:Thm. 3 in \cite{anshelevich2014approximating}, & (LB:Lm. \ref{lemma:rand-lb-2-non-strategic}, & (LB:Lm. \ref{obs:lb-location-2}, \\
						& UB:Lm. \ref{lemma:det-med-3}) & UB:Lm.   \ref{lemma:det-med-3}) & UB:Thm. \ref{thm-spike} ) & UB:Thm.\ref{thm-spike}) \\
						
						\hline 
						Ranking & 3 & 3 & 2 & 2 \\
						& (LB:Lm. \ref{lemma:lb-det-3}, & (LB: Thm. 3 in \cite{anshelevich2014approximating}, & (LB:Lm. \ref{lemma:rand-lb-2-non-strategic}, & (LB:Obs. \ref{obs:lb-location-2}, \\
						& UB:Lm. \ref{lemma:det-med-3}) & UB: Lm. \ref{lemma:det-med-3}) & UB:Thm. \ref{thm-spike} ) & UB:Lm. \ref{lemma:det-med-3})

%
						\\ 
						
						\hline 
						Location & 3 & 1 & 2 & 1 \\
						(high information) & (LB:Lm. \ref{lemma:lb-det-3}, & & (LB:Lm. \ref{lemma:rand-lb-2-non-strategic}, & \\
						&  UB:Lm. \ref{lemma:det-med-3}) & & UB:Thm. \ref{thm-spike} ) & \\
						
						
						\hline 
						\end{tabular} 
						
						
						%
						\caption{The approximation ratios of mechanisms on the line ($\Re$) in various settings.
							All the results in the table are tight.
							In particular the upper bound in strategic case is due to the spike mechanism.
						}
						
						\label{fig:results-table-line}

			\end{minipage}}} \par\setlength{\fboxrule}{0.2pt}								
					
		\end{figure}
		
	\end{itemize}

\subsection{Related Work}

Voting systems have been a domain of prolific research for decades.
The seminal Gibbard-Satterthwaite theorem \cite{gibbard1973manipulation} shows that if the rankings of agents can be arbitrary and the amount of candidates is greater than two, then the only onto truthful mechanisms are dictatorships.
However, if there are limitations on the rankings, then the impossibility theorem of Gibbard-Satterthwaite does not hold.
In many cases the rankings can be limited to single-peaked preferences, a notion used as early as 1948 by Black \cite{black1948rationale}.
In 1980 Moulin showed a complete characterization of truthful deterministic mechanisms for single-peaked preferences \cite{moulin1980strategy}.
Schummer and Vohra \cite{schummer2002strategy} extended this characterization to cycles and general graphs.

There has been extensive work describing various candidate selection mechanisms, which have been generally divided to 3 main types (\cite{brandt2012computational}, \cite{walsh2007uncertainty}) --- scoring rules (e.g., plurality, Borda, anti-plurality, range voting, cumulative), Condorcet extensions (e.g., Copeland, maxmin, Dodgson, Young, ranked pairs), or other mechanisms (e.g., single transferable vote, Bucklin).
Some work on social choice also made use of randomized voting schemes, for instance in order to improve the results of the mechanisms \cite{procaccia2010can} or to make manipulation computationally hard (\cite{conitzer2006nonexistence} pages 632-633).
Most of these mechanisms have no assumptions on the preferences of the agents, and are rankings mechanisms (i.e., they receive the ordinal preferences of the voters as input). 
Since in these circumstances the Gibbard-Satterthwaite impossibility theorem holds, the mechanisms are typically not truthful. 
While the importance of truthfulness and of simple mechanisms (with less options for each voter) has been acknowledged, to the best of our knowledge there has not been a formal framework for reduction or an assessment of the relationships between the different types of mechanisms.

Since in the lack of cardinal costs no global objective functions can be measured (e.g, the social cost), the focus of many of the aforementioned mechanisms is on achieving some desirable axiomatic properties.
Nonetheless, the use of utilitarianism in the realm of social choice has firm and ancient roots (see, e.g, a 1952 paper by Fleming \cite{fleming1952cardinal} and a 1955 work by Harsanyi \cite{harsanyi1955icardinal}).
Moreover, a new line of work commenced in recent years regarding {\em distortion}, which also refers to the utilitarian goal of minimizing social cost (partly due emergence of new domains such as recommender systems and e-commerce, as previously noted). 
The term was coined in 2006 by Procaccia and Rosenschein \cite{procaccia2006distortion}, and it was later used, for instance, by Boutilier et. al. \cite{boutilier2012optimal}.
Recently, Anshelevich et al. \cite{anshelevich2014approximating} assessed the distortion of several voting rules, and provided lower bounds on them.
In \cite{anshelevich2014approximating}, the distortion is the worst case ratio between the social cost of the candidate elected and the social cost of the optimal candidate, over any ranking profile (that is, preference profile) in any metric space.
The distortion is a quite similar to the approximation ratio used in this paper, but it differs in two key properties --
\begin{itemize}
	\item Most importantly, the source of the imperfection depicted by the distortion is the mechanism's lack of information (the mechanism has access to the ordinal ranking of the agents, but not to their exact location, that is - not to their full cardinal utilities). In this paper, the approximation ratio is greater than one {\sl both} because of this information deficiency (in the cases or ranking and voting mechanisms), {\sl and} because of incentive compatibility constraints.
	In this sense, we can quantify the cost of limited information as well as the cost of truthfulness in various settings.
	\item The distortion calculates the worst-case ratio in any metric space, whereas the approximation ratio is sometimes calculated over a specific metric space.
\end{itemize}

Anshelevich et al. show a deterministic lower bound of 3 on the distortion, and they prove that two mechanisms (social choice functions), Copeland and Uncovered Set, achieve a distortion of 5.

Procaccia and Tennenholtz introduced game theoretical aspects to the facility location problem.
As mentioned before, their setting is similar to the one in this paper, except that the location of the facility is not restricted to a set of candidates, but instead can be located at any point on the line.
This model was extended by these authors and by others in many different ways.
The metric space researched spanned from a line (\cite{fotakis2014power}, \cite{procaccia2009approximate}) to a circle (\cite{alon2009strategyproof}, \cite{alon2010walking}), a tree (\cite{alon2009strategyproof}, \cite{feldman2013strategyproof}) or a general graph (\cite{alon2009strategyproof}).
There are many papers regarding building several facilities (or electing a committee of candidates), where the cost of an agent is her distance to the closest facility (\cite{fotakis2013strategyproof}, \cite{fotakis2014power}, \cite{lu2010asymptotically}, \cite{procaccia2009approximate}).
As opposed to the voting scenario, the goal of the vast majority of these papers was to optimize over some global target function, and the most popular target functions were the utilitarian (social cost) and egalitarian (the maximal cost of an agent) (see, e.g, \cite{alon2009strategyproof}, \cite{fotakis2014power}, \cite{procaccia2009approximate}), but there were also works regarding additional target functions like the $L_2$ norm (the sum of the squared distances of the agents, see \cite{feldman2013strategyproof}).
Some papers consider ``obnoxious facility location" --- a setting in which agents want to be as far away as possible from the facility, e.g., when selecting a location for a central garbage dump (\cite{cheng2013strategy}).

When the outcome is constrained to a set of candidates, the facility location literature is far less extensive.
In this setting, a recent paper by Sui and Boutilier defines a set of deterministic mechanisms which is GSP on the line and $\epsilon$-GSP on $\Re^n$ \cite{sui2015approximately}.
The paper does not show bounds on global objectives such as the social cost.

Dokow et al. \cite{dokow2012mechanism} characterize deterministic truthful mechanisms on the discrete line and the discrete circle. They move on to give a lower bound on the social cost for large circles, and to the best of our knowledge this is the only result regarding assessment of the social cost in a constrained setting.
It is worthy to note the model in \cite{dokow2012mechanism} has 2 major properties which differ from the one in this paper: (a) The discrete constraints of the locations apply to the agents as well as to the candidates, so all agents are located precisely on some candidate; (b) The distance between any two neighboring candidates must be constant (for instance: 1,2,3,4,...).

%% file: model_sv.tex
\section{Model} \label{model}

Let $N=\{1 \ldots n\}$ be a set of agents, where each agent $i \in N$ is located at some point $x_i$.
We refer to the location of agent $i$ as agent $i$'s {\em type}.
Let $\bx=(x_1 \ldots x_n)$ be the {\em location profile} of the agents.
There exists a fixed set of candidates $\mathcal{C}=\{\mathcal{C}_1 \ldots \mathcal{C}_m\}$. 
Each candidate $\mathcal{C}_j$, is located at point $y_j$, and this location is publicly known. 
The agents and candidates are located on some metric space.
A significant part of the paper deals with specific metric spaces, and these will be specifically noted.
In the parts where the metric space is $\Re$, it is assumed that the agents and the candidates are both numbered in ascending order based on their locations (otherwise they could be renamed in this manner).

A {\em deterministic mechanism} $M$, is a function which maps an {\em action profile} $\bara=\{a_1 \ldots a_n\} \in \mathcal{A}^n$ to a candidate, that is: $M: \mathcal{A}^n \rightarrow \mathcal{C}$.
We consider three classes of mechanisms that differ in the input they accept, {\sl i.e.}, in the action space $\mathcal{A}$ of the agents:
\begin{itemize}
	\item {\em Voting mechanisms}, in which each agent casts a vote for a candidate, that is: $a_i \in \mathcal{C}$. 
	\item {\em Ranking mechanisms}, in which every agent reports ordinal preferences over all the $m$ candidates. The notation $\mathcal{C}_j \succeq \mathcal{C}_k$ indicates a preference of candidate $\mathcal{C}_j$ over candidate $\mathcal{C}_k$ (or is indifferent between the two). In ranking mechanisms $a_i \in \Pi$, where $\Pi$ is the set of all permutations of the set of candidates $\mathcal{C}$. These mechanisms are sometimes referred to in the literature as social choice functions. 
	\item {\em Location mechanisms}, in which every agent reports their location, that is $a_i$ is some point in the metric space.
\end{itemize}

Given a joint action profile $\bara$, the \textit{cost} of point $x$ is its distance to the facility, that is: $\cost_{x}(M,\bara)=|x-M(\bara)|$. For agent $i \in N$ located at point $x_i$, we refer to $\cost_{x_i}(M,\bara)$ as the {\em cost of agent $i$}.
The goal of each agent is to minimize her cost.

Truthful mechanisms are usually defined in the context of direct revelation mechanisms. 
Since in ranking and voting mechanisms the action space does not coincide with the type space, we extend this notion in the following trivial manner for these cases as well. 
For an agent in location $x_i$ and for any mechanism (location, ranking or voting), let $\cala(x_i)$ be the set of {\em true actions} of this agent --- the actions which convey the real preferences of this agent. 
For instance, in voting mechanisms $\cala(x_i)$ is the set of candidates closest to $x_i$, which we refer to as the {\em favorite candidates} of $x_i$ (this might be a set since there may be ties). 
An agent reporting $a_i \in \cala(x_i)$ is said to be reporting {\em truthfully}, and an action profile $\bara$ in which all agents report truthfully is called a {\em truthful profile}.
The set of truthful profiles is denoted $\cala(\bx)$.
A {\em truthful mechanism} $M$ is one in which no agent can suffer from reporting truthfully, regardless of the actions of the other agents: 
$$\forall i \in N, \forall x_i, \forall a_i \in \cala (x_i), \forall a_{-i} \in \cala^{n-1}, \forall a_i' \in \cala: \cost_{x_i}\left(M,(a_i,a_{-i}) \right) \leq \cost_{x_i}\left(M,(a_i',a_{-i}) \right)$$

A {\em randomized mechanism} is a mapping from an action profile to a distribution over the candidates, that is: $M: \cala^n \rightarrow \Delta(\calc)$.
The cost of agent $i$ is the expected cost of this agent according to the probability distribution returned by the mechanism, that is: $\cost_{x_i}(M, \bara)=\mathbb{E}_{\calc_j\sim M(\bara)}|x_i-y_j|$.

Two different notions of randomized truthful mechanisms have been studied in the literature, and we extend them naturally based on our definitions of truthful reports: 
\begin{itemize}
	\item {\em Truthful in expectation (TIE)} mechanisms --- where the expected cost of an agent reporting truthfully is never higher than any other action.
	That is: $\forall i \in N$, $\forall a_i \in \cala (x_i)$, $\forall a_{-i} \in \cala^{n-1}$, $\forall a_i' \in \cala$: $\cost_{x_i}\left(M,(a_i,a_{-i}) \right) \leq \cost_{x_i}\left(M,(a_i',a_{-i}) \right)$.
	In these mechanisms the agent may regret her action ex-post for some of the instances. 
	\item {\em Universally truthful} mechanisms are mechanisms which can be expressed as a probability distribution over deterministic truthful mechanisms. In these mechanisms an agent never regrets reporting truthfully, even after the random outcome is unraveled. 
\end{itemize}
Clearly, every universally truthful mechanism is truthful in expectation mechanisms, but not necessarily vice versa. 
Throughout the paper, in the randomized setting we use the term ``truthful" to refer to truthful in expectation mechanisms, unless otherwise stated. 

The {\em social cost} of a mechanism is the sum of the agents' costs. 
For a location profile $\mathbf{x}$ and an action profile $\bara$ the social cost is: $\SC(M,\mathbf{x}, \bara) = \sum_i \cost_{x_i}(M, \bara)$. 
The cost of a candidate is the cost of the mechanism which locates the facility on that candidate, that is: $\SC(\calc_j,\bx)= \sum_{i \in N}|y_j-x_i|$.
Given a location profile $\mathbf{x}$, the {\em optimal mechanism}, denoted $\OPT(\mathbf{x})$, is one which chooses a candidate that minimizes the social cost ($\calc_{opt}$). 
For the sake of consistency, when there are several optimal candidates, we refer to the leftmost among them as $\calc_{opt}$.
For any truthful in expectation mechanism $M$ (including universally truthful mechanisms), the {\em social cost of $M$ given a location profile $\bx$} is the maximal social cost it yields by any truthful action profile $\bara$, that is: $\SC(M,\bx) = \max_{\bara \in \cala(\bx)}\SC (M,\bx,\bara)$.
The {\em approximation ratio} of a truthful in expectation mechanism $M$ is the maximal ratio for any location profile $\bx$, between social cost of $M$ given $\bx$ and the optimal social cost given $\bx$: $\max_{\bx}\frac{\SC(M,\bx)}{\SC(\OPT,\bx)}$.

We make use of several terms which are relevant for voting mechanisms:
\begin{itemize}
	\item The location of a vote $a_i$ (some candidate) is denoted $y(a_i)$.
	\item When the network is the line, then there is an inherent order of the votes, and therefore it is possible to make use of percentiles. 
	A {\em percentile mechanism} is a voting mechanism that elects the $i$'th percentile vote (for example, the mechanism which chooses the leftmost vote is the $0$ percentile mechanism). 
	\item A {\em weighted percentile voting (WPV)} mechanism locates the facility on the $i$'th percentile vote with some probability $p_i$, where $\{p_i\}$ does not depend on the action profile $\bara$. For example, ``random dictator" is a WPV mechanism which chooses any vote $a_i$ with probability $\frac{1}{n}$.
\end{itemize}

In voting mechanisms, the set of candidates $\calc$ induces a partition of the line in the following manner --- the {\em voting zone} of candidate $\calc_i$, denoted $\calv_i$, is the set of points whose favorite candidate is $\calc_i$: $\calv_i=\{x: \forall \calc_j: |x- \calc_i| \leq |x-\calc_j| \}$.
The voting zones are bounded by {\em voting borders}.
For example, when the metric space is $\Re$, there are $n-1$ borders, which are the midpoints between two consecutive candidates: $b_i = \frac{y_i+y_{i+1}}{2}$ (see Figure \ref{fig:definitions-voting}). 
When the metric space is $\Re^d$, the voting zones create a Voronoi diagram.
A candidate which has at least one agent in their zone is called {\em active}.

In ranking mechanisms, $\calc$ induces a partition which divides the line into {\em ranking zones}.
All points in some ranking zone $\calr_i$ share some ranking $\pi_i$. 
In this case, we say that the ranking $\pi_i$ is consistent with ranking zone $\calr_i$. 
The ranking zones bounded by {\em ranking borders}.
For example, when the metric space is $\Re$ the ranking borders are the midpoints between any two candidates: $b_{i,j} = \frac{y_i+y_{j}}{2}$. 

\begin{figure}[t]
	
	\setlength{\fboxrule}{0.5 pt}
	\noindent \fbox{\noindent\makebox[1\textwidth][c]
		{\begin{minipage}{1\textwidth}
			
			\centering		

			\begin{tikzpicture}[y=.3cm, x=.3cm,font=\sffamily]			
			\draw[->, thick] (0,0) -- (40,0) node{};
			\draw[fill=white] (5, 0) circle (0.8) node [above, yshift=0.3cm] {$y_1$};
			\draw[fill=white] (15, 0) circle (0.8) node [above, yshift=0.3cm] {$y_{2}$};
			\draw[fill=white] (34, 0) circle (0.8) node [above, yshift=0.3cm] {$y_{3}$};
			
			\draw (10, 1) -- (10, -1) node [above, , yshift=0.5cm, align=center]{$b_1$};
			\draw (24.5, 1) -- (24.5, -1) node [above, yshift=0.5cm, align=center]{$b_{2}$};
			
			\draw[fill=black] (3, 0) circle (0.4) node [below, yshift=-0.3cm] {$x_1$};
			\draw[fill=black] (12, 0) circle (0.4) node [below, yshift=-0.3cm] {$x_{2}$};
			\draw[fill=black] (23, 0) circle (0.4) node [below, yshift=-0.3cm] {$x_{3}$};
			\draw[fill=black] (27, 0) circle (0.4) node [below, yshift=-0.3cm] {$x_{4}$};
			\draw[fill=black] (30, 0) circle (0.4) node [below, yshift=-0.3cm] {$x_{5}$};			
			
			\draw[dim] (0,3) -- (10,3) node[midway,above] {$\calv_{1}$};
			\draw (10, 2.7) -- (10, 3.3) node [above]{};	 				
			\draw[dim] (10,3) -- (24.5,3) node[midway,above] {$\calv_{2}$};
			\draw (24.5, 2.7) -- (24.5, 3.3) node [above]{};	 	
			\draw[dim] (24.5,3) -- (40,3) node[midway,above] {$\calv_{3}$};   
			
			\draw[dim] (5,-3) -- (10,-3) node[midway,above] {};
			\draw (7.5, -2.7) -- (7.5, -3.3) node [above]{};	 				
			\draw[dim] (10,-3) -- (15,-3) node[midway,above] {};
			\draw (12.5, -2.7) -- (12.5, -3.3) node [above]{};	 	
			
			\end{tikzpicture}
			
		\caption{Illustration of candidates (white circles), agents (black circles), voting borders and voting zones when the metric space is $\Re$. 
		\newline
		For example, in this case the favorite candidate of $x_3$ is $\calc_2$: $\calc(x_{3})=\{\calc_2\}$. 
		The voting borders divide the distance between two consecutive candidates exactly in half, for example: $|b_1-y_1|=|y_2-b_1|$.}			
		\label{fig:definitions-voting}
		\end{minipage}}} \par\setlength{\fboxrule}{0.2pt}

\end{figure}

%% file: classes_sv.tex
\section{Classes of Mechanisms} \label{sec-classes}

In this section we go over the containment hierarchy of various classes of truthful mechanisms (e.g., Figure \ref{fig:classes}). 
We start with some intuition, then defining some necessary terms, and finally present the main theorem of this section.

Intuitively, for any mechanism $M$, there exists a mechanism $M'$ which receives a ``richer" input than $M$, and acts identically to $M$. 
For instance, for some arbitrary voting mechanism $M$, there obviously exists a ranking mechanism $M'$ which disregards all of the preferences except the top choice of each agent, and behaves essentially just like $M$ does.

We generalize this notion in the following informal definition --- 
a mechanism $M$ (whether location/ranking/or voting) is said to be {\em reducible} to a mechanism $M'$ (location/ranking/or voting) if for every location profile $\bx$ and true reports, the output of $M$ is identical to the output of $M'$ (a formal definition, which is based on $M$ simulating $M'$, is deferred to appendix \ref{def:reduciblity}). 

As pointed out, it is clear that every voting mechanism $M$ is reducible to some ranking mechanism $M'$ (or some location mechanism $M'$).
In these cases, if $M$ is truthful then so is $M'$, since $M'$ only uses the information which is inputted to $M$, so any misreports to $M'$ which would not change the input of $M$ do not affect the outcome at all. 
Note that the same reasoning also shows that every ranking mechanism is reducible to some location mechanism, and that any voting mechanism is reducible to some location mechanism.

On the other hand, it is not true that every ranking mechanism is reducible to some voting mechanism. 
Somewhat surprisingly, we will soon show that when we restrict ourselves to deterministic truthful mechanism this does hold, that is --- every {\sl deterministic truthful} ranking mechanism is reducible to some deterministic truthful voting mechanism. 

Two sets of mechanisms, $A$ and $B$, are said to be {\em equivalent} if every $a\in A$ is reducible to some $b\in B$, and every $b\in B$ is reducible to some $a\in A$. 

A set of mechanisms $A$ is said to be {\em strictly contained} in a set of mechanisms $B$ if every mechanism $a\in A$ is reducible to some mechanism $b \in B$, yet not every mechanism $b\in B$ is reducible to some mechanism $a\in A$. 
This is a slight abuse of terminology since the sets $A$ and $B$ may be disjoint, as their input space may be different.

The following theorem shows several claims regarding relations (equivalence or strict containment) between sets of truthful mechanisms. 
Notice that not only does this theorem show the hierarchy of the different classes, but it also provides notions relevant to characterization of truthful mechanisms. 
For instance, the second claim proves that no mechanism can use any information regarding the location of the agents beyond their ranking, while maintaining truthfulness.
In addition, in the claims showing strict containment, the examples in the proofs portray the expressiveness that the additional information gives the mechanism.

\begin{theorem} \label{thm:classes}
	The following claims hold in the Euclidean metric space $\Re^d$ (for any $d \in \mathbb{N}$):
	\begin{enumerate}
		\item The set of truthful deterministic ranking mechanisms strictly contains the set of truthful deterministic voting mechanisms.
		\item The set of truthful deterministic location mechanisms is equivalent to the set of truthful deterministic ranking mechanisms.
		\item The set of truthful in expectation randomized ranking mechanisms strictly contains the set of truthful in expectation randomized voting mechanisms.
		\item The set of truthful in expectation randomized location mechanisms strictly contains the set of truthful in expectation randomized ranking mechanisms.
		\item The set of truthful in expectation randomized voting mechanisms strictly contains the set of universally truthful randomized voting mechanisms.
		\item When there are two candidates, the set of truthful in expectation randomized location mechanisms is equivalent to the set of truthful in expectation randomized voting mechanisms.
	\end{enumerate}
\end{theorem}

For the ease of readability, we defer the proof of this theorem to the appendix (see \ref{prf:classes}).

%% file: spike_sv.tex
\section{Spike Mechanism} \label{sec-spike}
	
	In the next sections we will prove that both the median mechanism and the random dictator mechanism achieve an approximation ratio of three on $\Re$. 
	However, the cause for this ratio in these two cases is different - for median it is due to an instance which is bad for the median agent, while for random dictator it is due to a bad instance for an agent in one of the extremes. 
	It is therefore desirable to devise a mechanism which is resistant to bad instances of any agent.
	The spike mechanism stems from this intuition.
	
	This section contains foundations needed for the introduction of the spike mechanism, the definition of the mechanism, and the theorem showing that spike achieves an approximation ratio of 2. In the entirety of this section, the metric space is $\Re$ and the mechanisms are voting mechanisms. 
	
	We start by showing a basic lemma regarding WPV mechanisms. Recall that these are voting mechanisms which choose the $i$'th percentile vote with a predetermined probability $p_i$.
	
	\begin{lemma} \label{lemma:WPV-truthful}
		Any weighted percentile voting mechanism $M$ is universally truthful.
	\end{lemma}	
	The proof is straightforward and deferred to appendix \ref{prf:wpv}.
	
	\begin{definition}[Cumulative count of a candidate]
		Given a voting profile $\bara$, the cumulative count of candidate $\calc_i$ is the amount of agents who voted for any candidate who is not located to the right of $y_i$, that is: $t(i)=|j \in N:y(a_j)\leq y_i|$. 
	\end{definition}
	
	\begin{definition}[Spike Mechanism] \label{def:spike}
		The spike mechanism is a WPV mechanism \footnote[3]{Spike is a WPV mechanism since it chooses the $i$'th percentile vote with some probability $p_i$, and $p_i$ does not depend on the reports.} that for each voting profile $\bara$, chooses candidate $\calc_i$ according to the following cumulative distribution function:
		\begin{eqnarray*}				
			F (i) = \begin{cases}
				\frac{t(i)}{2(n-t(i))} & \quad \text{if }  t(i) \leq n/2 \\
				1.5 - \frac{n}{2t(i)} & \quad \text{if } t(i) > n/2
			\end{cases}
		\end{eqnarray*}
	\end{definition}
	The mechanism is named after the shape of the density function it creates (see Figure \ref{fig:spike-prob}).
	Recognize that the result of the mechanism depends on the amount of votes that each candidate received and on the order of the candidate along the line, but not on the distances between the candidates. 
	
	\begin{observation}
		Spike defines a symmetric distribution, that is: $\forall i: F(i)=1-F(n-i)$.
	\end{observation}
	\begin{proof}
		Let $1 \leq i \leq n/2$, then:
		\begin{eqnarray*}		
			&& F(i)=\frac{i}{2(n-i)}
		\end{eqnarray*}
		\begin{eqnarray*}
			1-F(n-i)&=& 1-\left(1.5-\frac{n}{2(n-i)}\right)=\frac{n}{2(n-i)}-\frac{1}{2}=\frac{n-(n-i)}{2(n-i)}= \frac{i}{2(n-i)}
		\end{eqnarray*}	
	\end{proof}
	
	We now define a few terms needed for the proof of the approximation ratio. Recall that $\calc_{opt}$ is uniquely defined for a location profile $\mathbf{x}$, since ties are broken in favor of the leftmost candidate.
	We denote the set of borders $\{b_i\}_{i=1}^{m-1}$ by $B$.	
	
	\begin{definition} [Tight profile of $\mathbf{x}$, see Figure \ref{fig:example-tight-before}] 
		Given a location profile $\mathbf{x}$, the profile $\mathbf{x}'$ is said to be the \textit{tight profile of $\mathbf{x}$} if it moves all agents who are not on a border as close as possible to $\calc_{opt}$ within their zones, that is:
		\begin{eqnarray*}		
			\forall i: x_i' = \begin{cases}
				x_i & \quad \text{if }  x_i \in B \\
				y_{opt} & \quad \text{if }  x_i \in \calv_{opt} \setminus B \\
				b_j & \quad \text{if } x_i \in \calv_j \setminus B \text{ and } j < \OPT \\
				b_{j-1} & \quad \text{if } x_i \in \calv_j \setminus B \text{ and } j > \OPT 
			\end{cases}
		\end{eqnarray*}	
	\end{definition}
	
		\begin{figure}[t]
		
		\setlength{\fboxrule}{0.5 pt}
		\noindent \fbox{\noindent\makebox[1\textwidth][c]
			{\begin{minipage}{1\textwidth}					
					\centering					
					\caption{Tight profile}	
					
					\begin{subfigure}{1\textwidth}
						\centering						
						\begin{tikzpicture}[y=.3cm, x=.3cm,font=\sffamily]
						
						\draw[-, thick] (0,0) -- (35,0);
						
						\draw[fill=white] (4, 0) circle (0.5) node [below, yshift=-0.3cm] {$y_1$};
						\draw[fill=white] (14, 0) circle (0.5) node [below, yshift=-0.3cm] {$y_2=y_{opt}$};	
						\draw[fill=white] (22, 0) circle (0.5) node [below, yshift=-0.3cm] {$y_3$};
						\draw[fill=white] (30, 0) circle (0.5) node [below, yshift=-0.3cm] {$y_4$};
						
						\draw (9, -2) -- (9, 2) node [above, align=center]{$b_1=IB(x_1)$};
						\draw (18, -2) -- (18, 2) node [above, align=center]{$b_2=IB(x_3)$};
						\draw (26, -0.5) -- (26, 2) node [above, align=center]{$b_3$};				
						
						\draw[fill=black] (2, 0) circle (0.3) node [below, yshift=-0.3cm] {$x_1$};
						\draw[fill=black] (10, 0) circle (0.3) node [below, yshift=-0.3cm] {$x_2$};
						\draw[fill=black] (20, 0) circle (0.3) node [below, yshift=-0.3cm] {$x_3$};
						\draw[fill=black] (26, 0) circle (0.3) node [below, yshift=-0.3cm] {$x_4$};
						\end{tikzpicture}	
						
						\caption{The original profile $\bx$. }
						\label{fig:example-tight-before}
					\end{subfigure}

					\begin{subfigure}{1\textwidth}
						\centering		
						
						\begin{tikzpicture}[y=.3cm, x=.3cm,font=\sffamily]
						
						\draw[-, thick] (0,0) -- (35,0);
						
						\draw[fill=white] (4, 0) circle (0.5) node [below, yshift=-0.3cm] {$y_1$};
						\draw[fill=white] (14, 0) circle (0.5) node [below, yshift=-0.3cm] {$y_2=y_{opt}$};	
						\draw[fill=white] (22, 0) circle (0.5) node [below, yshift=-0.3cm] {$y_3$};
						\draw[fill=white] (30, 0) circle (0.5) node [below, yshift=-0.3cm] {$y_4$};
						
						\draw (9, -0.5) -- (9, 2) node [above, align=center]{$b_1$};
						\draw (18, -0.5) -- (18, 2) node [above, align=center]{$b_2$};
						\draw (26, -0.5) -- (26, 2) node [above, align=center]{$b_3$};				
						
						\draw[fill=black] (9, 0) circle (0.3) node [below, yshift=-0.3cm] {$x_1$};
						\draw[fill=black] (14, 0) circle (0.3) node [below, yshift=0.7cm] {$x_2$};
						\draw[fill=black] (18, 0) circle (0.3) node [below, yshift=-0.3cm] {$x_3$};
						\draw[fill=black] (26, 0) circle (0.3) node [below, yshift=-0.3cm] {$x_4$};
						
						\draw[->, thick] (2.5,0.5) .. controls (5.75,1) .. (8.5, 0.5);	
						\draw[->, thick] (10.5,0.5) .. controls (12,1) .. (13.5, 0.5);	
						\draw[->, thick] (19.75,0.5) .. controls (19,0.75) .. (18.25, 0.5);					
						\end{tikzpicture}
						\caption{The tight profile of $\bx$. The arrows represent the movements between the two profiles. }
						\label{fig:example-tight-after}		
					\end{subfigure}	
					
				\end{minipage}}} \par\setlength{\fboxrule}{0.2pt}				
				
			\end{figure}
			
			\begin{definition} [Left-compressed profile of $\mathbf{x}$, see Figure \ref{fig:example-lc-before}] 
				Given a tight location profile $\mathbf{x}$, a left-compressed profile of $\mathbf{x}$ moves all the agents on the leftmost border to their neighboring border on the right, if this border is left of $y_{opt}$. Formally: let the location of the leftmost agent be $x_1=b_j$, then the left-compressed profile of $\mathbf{x}$ is:
				\begin{eqnarray*}		
					\forall i: x_i' = \begin{cases}
						b_{j+1} & \quad \text{if }  (x_i = b_j) \wedge (b_{j+1} < y_{opt}) \\
						x_i & \quad \text{otherwise.}
					\end{cases}
				\end{eqnarray*}		
			\end{definition}
			
			Note that the left compressed profile of a tight profile is also a tight profile. 
			The right-compressed profile of $\mathbf{x}$ is defined in a completely symmetrical fashion.

			\begin{figure}[t]
				
				\setlength{\fboxrule}{0.5 pt}
				\noindent \fbox{\noindent\makebox[1\textwidth][c]
					{\begin{minipage}{1\textwidth}
							
							\centering	
							\caption{Left-compressed profile}
							\begin{subfigure}{1\textwidth}
								\centering
								
								\begin{tikzpicture}[y=.3cm, x=.3cm,font=\sffamily]
								
								\draw[-, thick] (0,0) -- (35,0);
								
								\draw[fill=white] (31, 0) circle (0.5) node [below, yshift=-0.3cm] {$y_5$};
								\draw[fill=white] (21, 0) circle (0.5) node [below, yshift=0.7cm] {$y_4=y_{opt}$};	
								\draw[fill=white] (15, 0) circle (0.5) node [below, yshift=-0.3cm] {$y_3$};
								\draw[fill=white] (8, 0) circle (0.5) node [below, yshift=-0.3cm] {$y_2$};
								\draw[fill=white] (3, 0) circle (0.5) node [below, yshift=-0.3cm] {$y_1$};
								
								\draw (26, -0.5) -- (26, 2) node [above, align=center]{$b_4$};
								\draw (18, -0.5) -- (18, 2) node [above, align=center]{$b_3$};
								\draw (11.5, -0.5) -- (11.5, 2) node [above, align=center]{$b_2$};				
								\draw (5.5, -0.5) -- (5.5, 2) node [above, align=center]{$b_1$};								
								
								\draw[fill=black] (26, 0) circle (0.3) node [below, yshift=-0.3cm] {$x_5$};	
								\draw[fill=black] (21, 0) circle (0.3) node [below, yshift=-0.3cm] {$x_4$};
								\draw[fill=black] (18, 0) circle (0.3) node [below, yshift=-0.3cm] {$x_3$};
								\draw[fill=black] (11.5, 0) circle (0.3) node [below, yshift=-0.3cm] {$x_2$};
								\draw[fill=black] (5.5, 0) circle (0.3) node [below, yshift=-0.3cm] {$x_1$};

								\end{tikzpicture}	
								
								\caption{$\mathbf{x}$, the original tight profile.}
								\label{fig:example-lc-before}
							\end{subfigure}
							\begin{subfigure}{1\textwidth}
								\centering
								
								\begin{tikzpicture}[y=.3cm, x=.3cm,font=\sffamily]
								
								\draw[-, thick] (0,0) -- (35,0);
								
								\draw[fill=white] (31, 0) circle (0.5) node [below, yshift=-0.3cm] {$y_5$};
								\draw[fill=white] (21, 0) circle (0.5) node [below, yshift=0.7cm] {$y_4=y_{opt}$};	
								\draw[fill=white] (15, 0) circle (0.5) node [below, yshift=-0.3cm] {$y_3$};
								\draw[fill=white] (8, 0) circle (0.5) node [below, yshift=-0.3cm] {$y_2$};
								\draw[fill=white] (3, 0) circle (0.5) node [below, yshift=-0.3cm] {$y_1$};
								
								\draw (26, -0.5) -- (26, 2) node [above, align=center]{$b_4$};
								\draw (18, -0.5) -- (18, 2) node [above, align=center]{$b_3$};
								\draw (11.5, -0.5) -- (11.5, 2) node [above, align=center]{$b_2$};				
								\draw (5.5, -0.5) -- (5.5, 2) node [above, align=center]{$b_1$};								
								
								\draw[fill=black] (26, 0) circle (0.3) node [below, yshift=-0.3cm] {$x_5$};	
								\draw[fill=black] (21, 0) circle (0.3) node [below, yshift=-0.3cm] {$x_4$};
								\draw[fill=black] (18, 0) circle (0.3) node [below, yshift=-0.3cm] {$x_3$};
								\draw[fill=black] (11.5, 0) circle (0.3) node [below, yshift=-0.3cm] {$x_1,x_2$};
								
								\draw[->, thick] (6,0.5) .. controls (8.5,1) .. (11, 0.5);		
								\end{tikzpicture}
								
								\caption{The left-compressed profile of $\mathbf{x}$ (after $x_1$ moves to $b_2$).}
								\label{fig:example-lc-after}
							\end{subfigure}

						\end{minipage}}} \par\setlength{\fboxrule}{0.2pt}	
						
	\end{figure}
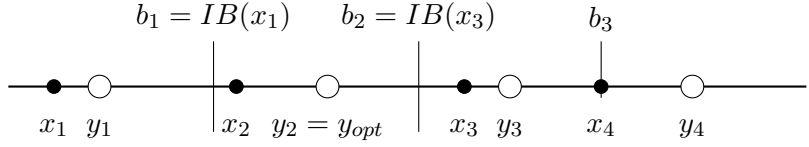
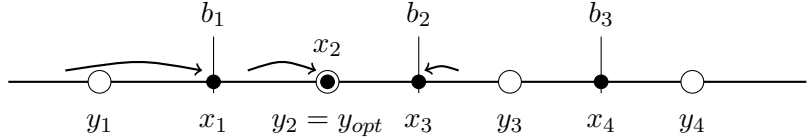	
	
	After compressing location profiles, there are likely to be locations in which there are many agents. 
	We therefore use the following notation: the location profile is written as $\mathbf{x}=\{(\hat{x}_1,n_1),...,(\hat{x}_k,n_k)\}$, which means that for each $j:1 \leq j \leq k$, there are $n_j$ agents located at $\hat{x}_j$. 		
	
	We now use these definitions to prove the main result of this section:
	
	\begin{theorem} \label{thm-spike}
		The spike mechanism is universally truthful, and it achieves an approximation ratio of 2 on $\Re$.
	\end{theorem}
	\begin{proof}
		Spike is a WPV mechanism, so according to Lemma \ref{lemma:WPV-truthful} it is universally truthful.
		
		The analysis of the approximation ratio is more involved and is based on backwards induction which follows these steps (see Figure \ref{fig:spike}):
		\begin{enumerate}
			\item Figure \ref{fig:spike-1}: Start with a general location profile $\mathbf{x}$, and compute its optimal candidate, $\calc_{opt}$.
			\item Figure \ref{fig:spike-2}: Let $\bx^{(1)}$ be the tight profile of $\mathbf{x}$. We show that the transition from $\bx$ to $\bx^{(1)}$ cannot decrease the approximation ratio (Lemma \ref{lemma:tight}).
			\item Figure \ref{fig:spike-3}: Let  $\mathbf{x}^{(2)}$ be the left-compression of  $\mathbf{x}^{(1)}$. We show that if the ratio of $\mathbf{x}^{(2)}$ is not higher than 2, then so is the ratio of $\mathbf{x}^{(1)}$ (Lemma \ref{lemma:inwards}).
			\item Figure \ref{fig:spike-4}: Repeat left and right compressions until we can no longer compress. At this stage, the profile is tight with at most 3 active candidates, and we note this profile $\bx^{(3)}$. We show that the approximation ratio of $\bx^{(3)}$ is not above 2 (Lemma \ref{lemma:three-candidates}).
		\end{enumerate}

		\begin{figure}[h!]
			
			\setlength{\fboxrule}{0.5 pt}
			\noindent \fbox{\noindent\makebox[1\textwidth][c]
				{\begin{minipage}{1\textwidth}
						
						\centering	
						\caption{Illustration of the proof of Theorem \ref{thm-spike}:}
						\begin{subfigure}{1\textwidth}
							\centering

							\begin{tikzpicture}[y=.3cm, x=.3cm,font=\sffamily]
							\draw[-, thick] (0,0) -- (40,0);

							\draw[fill=white] (1.5, 0) circle (0.5) node [above, yshift=0.3cm] {$y_1$};
							\draw[fill=white] (8, 0) circle (0.5) node [above, yshift=0.3cm] {$y_2$};	
							\draw[fill=white] (15, 0) circle (0.5) node [above, yshift=0.3cm] {$y_3$};
							\draw[fill=white] (21, 0) circle (0.5) node [above, yshift=0.3cm] {$y_4=y_{opt}$};
							\draw[fill=white] (29, 0) circle (0.5) node [above, yshift=0.3cm] {$y_5$};	
							\draw[fill=white] (38, 0) circle (0.5) node [above, yshift=0.3cm] {$y_{6}$};
							
							\draw (4.75, -2) -- (4.75, 2) node [above, align=center]{$b_1$};
							\draw (11.5, -2) -- (11.5, 2) node [above, align=center]{$b_2$};
							\draw (18, -0.5) -- (18, 2) node [above, align=center]{$b_3$};	
							\draw (25, -0.5) -- (25, 2) node [above, align=center]{$b_4$};	
							\draw (33.5, -2) -- (33.5, 2) node [above, align=center]{$b_5$};

							\draw[fill=black] (3, 0) circle (0.3) node [below, yshift=-0.3cm] {$x_1$};
							\draw[fill=black] (5.5, 0) circle (0.3) node [below, yshift=-0.3cm] {$x_2$};
							\draw[fill=black] (7, 0) circle (0.3) node [below, yshift=-0.3cm] {$x_3$};
							\draw[fill=black] (18, 0) circle (0.3) node [below, yshift=-0.3cm] {$x_4$};
							\draw[fill=black] (23, 0) circle (0.3) node [below, yshift=-0.3cm] {$x_{5}$};
							\draw[fill=black] (25, 0) circle (0.3) node [below, yshift=-0.3cm] {$x_{6}$};
							\draw[fill=black] (32, 0) circle (0.3) node [below, yshift=-0.3cm] {$x_{7}$};			
							\draw[fill=black] (37, 0) circle (0.3) node [below, yshift=-0.3cm] {$x_8$};
							
							\end{tikzpicture} 		
							
							\caption{Initial state: A general location profile $\bx$}
							\label{fig:spike-1}
						\end{subfigure}
						\begin{subfigure}{1\textwidth}
							\centering
							
							\begin{tikzpicture} [y=.3cm, x=.3cm,font=\sffamily]
							\draw[-, thick] (0,0) -- (40,0);

							\draw[fill=white] (1.5, 0) circle (0.5) node [above, yshift=0.3cm] {$y_1$};
							\draw[fill=white] (8, 0) circle (0.5) node [above, yshift=0.3cm] {$y_2$};	
							\draw[fill=white] (15, 0) circle (0.5) node [above, yshift=0.3cm] {$y_3$};
							\draw[fill=white] (21, 0) circle (0.5) node [above, yshift=0.3cm] {$y_4=y_{opt}$};
							\draw[fill=white] (29, 0) circle (0.5) node [above, yshift=0.3cm] {$y_5$};	
							\draw[fill=white] (38, 0) circle (0.5) node [above, yshift=0.3cm] {$y_{6}$};
							
							\draw (4.75, -0.5) -- (4.75, 2) node [above, align=center]{$b_1$};
							\draw (11.5, -0.5) -- (11.5, 2) node [above, align=center]{$b_2$};
							\draw (18, -0.5) -- (18, 2) node [above, align=center]{$b_3$};	
							\draw (25, -0.5) -- (25, 2) node [above, align=center]{$b_4$};	
							\draw (33.5, -0.5) -- (33.5, 2) node [above, align=center]{$b_5$};
							
							\draw[fill=black] (4.75, 0) circle (0.3) node (1){};
							\draw[fill=black] (11.5, 0) circle (0.3) node (2) {};
							\draw[fill=black] (18, 0) circle (0.3) node (3) {};
							\draw[fill=black] (21, 0) circle (0.3) node (4) {};
							\draw[fill=black] (25, 0) circle (0.3) node (5) {};
							\draw[fill=black] (33.5, 0) circle (0.3) node (6) {};
							
							\node [below = 0 cm of 1]{$\hat{x}_1$};
							\node [below = 0.4 cm of 1]{$n_1=1$};			
							\node [below = 0 cm of 2]{$\hat{x}_2$};
							\node [below = 0.4 cm of 2]{$n_2=2$};	
							\node [below = 0.2 cm of 3]{$\hat{x}_{3}$};
							\node [below = 0.7 cm of 3]{$n_{3}=1$};	
							\node [below = 0 cm of 4]{$\hat{x}_{4}$};
							\node [below = 0.3 cm of 4]{$n_{4}=1$};	
							\node [below = 0.2 cm of 5]{$\hat{x}_{5}$};
							\node [below = 0.7cm of 5]{$n_{5}=2$};	
							\node [below = 0 cm of 6]{$\hat{x}_{6}$};
							\node [below = 0.4 cm of 6]{$n_{6}=1$};				
							
							\draw[->, thick] (3.2,0.5) .. controls (3.9,0.7) ..(4.65,0.5);	
							\draw[->, thick] (5.75,0.5) .. controls (8.5,1) ..(11.25,0.5);
							\draw[->, thick] (7.25,0.5) .. controls (9.25,1) ..(11.25,0.5);
							\draw[->, thick] (22.8,0.5) .. controls (22,0.7) ..(21.2,0.5);	
							\draw[->, thick] (31.5,0.5) .. controls (28.5,1) ..(25.5,0.5);	
							\draw[->, thick] (36.5,0.5) .. controls (35.25,1) ..(34,0.5);	
							\end{tikzpicture}

							\caption{$\mathbf{x}^{(1)}$: The tight profile of $\mathbf{x}$, using the notation in which there are $n_i$ agents at point $\hat{x}_i$. 
							}							
							\label{fig:spike-2}
						\end{subfigure}	
						
						\begin{subfigure}{1 \textwidth}
							\centering

							\begin{tikzpicture}
							[y=.3cm, x=.3cm,font=\sffamily]
							\draw[-, thick] (0,0) -- (40,0);

							\draw[fill=white] (1.5, 0) circle (0.5) node (11)[above, yshift=0.3cm] {$y_1$};
							\draw[fill=white] (8, 0) circle (0.5) node (12)[above, yshift=0.3cm] {$y_2$};	
							\draw[fill=white] (15, 0) circle (0.5) node (13)[above, yshift=0.3cm] {$y_3$};
							\draw[fill=white] (21, 0) circle (0.5) node (14)[above, yshift=0.3cm] {$y_4=y_{opt}$};
							\draw[fill=white] (29, 0) circle (0.5) node (15)[above, yshift=0.3cm] {$y_5$};	
							\draw[fill=white] (38, 0) circle (0.5) node (16)[above, yshift=0.3cm] {$y_{6}$};
							
							\draw (4.75, -0.5) -- (4.75, 2) node [above, align=center]{$b_1$};
							\draw (11.5, -0.5) -- (11.5, 2) node [above, align=center]{$b_2$};
							\draw (18, -0.5) -- (18, 2) node [above, align=center]{$b_3$};	
							\draw (25, -0.5) -- (25, 2) node [above, align=center]{$b_4$};	
							\draw (33.5, -0.5) -- (33.5, 2) node [above, align=center]{$b_5$};
							
							\draw[fill=black] (11.5, 0) circle (0.3) node (2) {};
							\draw[fill=black] (18, 0) circle (0.3) node (3) {};
							\draw[fill=black] (21, 0) circle (0.3) node (4) {};
							\draw[fill=black] (25, 0) circle (0.3) node (5) {};
							\draw[fill=black] (33.5, 0) circle (0.3) node (6) {};

							\node [below = 0 cm of 2]{$\hat{x}_1$};
							\node [below = 0.4 cm of 2]{$n_1=3$};	
							\node [below = 0.2 cm of 3]{$\hat{x}_{2}$};
							\node [below = 0.7 cm of 3]{$n_{2}=1$};	
							\node [below = 0 cm of 4]{$\hat{x}_{3}$};
							\node [below = 0.3 cm of 4]{$n_{3}=1$};	
							\node [below = 0.2 cm of 5]{$\hat{x}_{4}$};
							\node [below = 0.7cm of 5]{$n_{4}=2$};	
							\node [below = 0 cm of 6]{$\hat{x}_{5}$};
							\node [below = 0.4 cm of 6]{$n_{5}=1$};		
							
							\draw[->, thick] (5.25,0.5) .. controls (8.1,1) ..(11,0.5);					
							\end{tikzpicture}

							\caption{$\mathbf{x}^{(2)}$: The left-compressed profile of $\mathbf{x}^{(1)}$ (moves $n_1$ agents from $b_1$ to $b_2$).}			
							
							\label{fig:spike-3}
						\end{subfigure}	
						
						\begin{subfigure}{1 \textwidth}
							\centering
							
							\begin{tikzpicture}
							[y=.3cm, x=.3cm,font=\sffamily]
							\draw[-, thick] (0,0) -- (40,0);
							
							\draw[fill=white] (1.5, 0) circle (0.5) node (11)[above, yshift=0.3cm] {$y_1$};
							\draw[fill=white] (8, 0) circle (0.5) node (12)[above, yshift=0.3cm] {$y_2$};	
							\draw[fill=white] (15, 0) circle (0.5) node (13){};
							\node [above=0cm of 13] {$y_L$};
							\node [above=0.5cm of 13]{$y_{3}$};			
							\draw[fill=white] (21, 0) circle (0.5) node (14){};
							\node [above=0.2cm of 14] {$y_C$};
							\node [above=0.7cm of 14]{$y_4=y_{opt}$};						
							\draw[fill=white] (29, 0) circle (0.5) node (15) {};	
							\node [above=0cm of 15] {$y_R$};
							\node [above=0.5cm of 15]{$y_{5}$};			
							
							\draw[fill=white] (38, 0) circle (0.5) node (16)[above, yshift=0.3cm] {$y_{6}$};
							
							\draw (4.75, -0.5) -- (4.75, 1.5) node [above, align=center]{};
							\draw (11.5, -0.5) -- (11.5, 1.5) node [above, align=center]{};
							\draw (18, -0.5) -- (18, 1.5) node [above, align=center]{$b_L$};	
							\draw (25, -0.5) -- (25, 1.5) node [above, align=center]{$b_C$};	
							\draw (33.5, -0.5) -- (33.5, 1.5) node [above, align=center]{};			
							
							\draw[fill=black] (18, 0) circle (0.3) node (4) {}; 		
							\node [below left= 0.6cm of 4] (5) {$\hat{x}_1$};  	
							\node [below= -0.1cm of 5] {$L=\displaystyle\sum_{i=1}^{opt-1}n_i=4$};
							\draw[->] (15,-2) -- (17.5,-0.5);
							\draw[fill=black] (21, 0) circle (0.3) node (6)[below, yshift=-0.2cm] {$\hat{x}_2$};
							\node [below= -0.2cm of 6] (7){$C=n_{opt}=1$};
							
							\draw[fill=black] (25, 0) circle (0.3) node (8) {}; 		
							\node [below right= 0.6cm of 8] (9) {$\hat{x}_3$};  	
							\node [below= -0.1cm of 9] {$R=\displaystyle\sum_{i=opt+1}^n n_i=3$};
							
							\draw[->] (28,-2) -- (25.5,-0.5);
							
							\draw (15, 4.8) -- (15, 5.2);	   	
							\draw[dim] (15,5) -- (18,5) node[midway,above] {$\beta$};      		
							\draw (18, 4.8) -- (18, 5.2);	   	
							\draw[dim] (18,5) -- (21,5) node[midway,above] {$\beta$};      	
							\draw (21, 4.8) -- (21, 5.2);	   	
							\draw[dim] (21,5) -- (25,5) node[midway,above] {$1$};      	
							\draw (25, 4.8) -- (25, 5.2);	   	
							\draw[dim] (25,5) -- (29,5) node[midway,above] {$1$};      		
							\draw (29, 4.8) -- (29, 5.2);	
							\end{tikzpicture}
							
							\caption{$\mathbf{x}^{(3)}$: The final profile after reapplying left and right compressions on $\mathbf{x}^{(2)}$ repeatedly. The active candidates are denoted $y_L$, $y_C$ and $y_R$, the amount of agents as $L$, $C$ and $R$ respectively, and we scale the distances to $\frac{b_C-y_C}{2}=1$ and $\frac{y_C-b_L}{2}=\beta$.}
														
							\label{fig:spike-4}
						\end{subfigure}		
						
						\label{fig:spike}		
					\end{minipage}}} \par\setlength{\fboxrule}{0.2pt}

				\end{figure}

				Proving these steps is sufficient to complete the proof of the theorem, since in Lemma \ref{lemma:three-candidates} we show that the ratio of  $\mathbf{x}^{(3)}$ is not higher than 2 (the base case). According to Lemma \ref{lemma:inwards}, this implies that the ratio of  $\mathbf{x}^{(1)}$ (prior to all of the compressions) is also not higher than 2 (the induction steps). Since the ratio of the $\mathbf{x}$  is not higher than that of $\mathbf{x}^{(1)}$, this means that the approximation ratio of $\mathbf{x}$ is not above 2, as needed.
				
				Notice that throughout this process $\calc_{opt}$ remains the optimal candidate, since it was optimal in the original profile $\mathbf{x}$, and in each step all agents move towards it, so the cost of any other candidate can decrease by no more than what the cost of $\calc_{opt}$ decreases. 
				
				Truthful reports to the spike mechanism (like any other voting mechanism) are not necessarily unique since an agent $i$ who is located on a border can report either of the two candidates closest to her. 
				For these cases of ties, we show that the worst-case ratio always occurs when the agents vote for the candidate located farther away from $\calc_{opt}$ (Lemma \ref{lemma:border-outward}, whose formal definition and proof are deferred to the appendix \ref{prf:tight}).
				
				We now present the lemmas formally. The proofs of the lemmas, which are given in the appendix, prove the backwards induction and conclude the proof of the theorem at large.
				
				\begin{lemma} \label{lemma:tight}
					Let $\mathbf{x}$ be an arbitrary location profile, let $\mathbf{x}'$ be the tight profile of $\mathbf{x}$ and let $M$ be an arbitrary WPV mechanism. Then the approximation ratio of $M$ given $\mathbf{x}'$ is not lower than that of $M$ given $\mathbf{x}$:
					\begin{eqnarray*}		
						\frac{\SC(M,\bx)}{\SC(\OPT,\bx)} 
						\leq \frac{\SC(M,\bx')}{\SC(\OPT,\bx')}.
					\end{eqnarray*}
				\end{lemma}	
				
				The previous lemma holds for any WPV mechanism, in particular for spike.			
				
				\begin{lemma} \label{lemma:inwards}
					Let $\mathbf{x}$ be a tight location profile, let $\mathbf{x}'$ be the left-compressed profile of $\mathbf{x}$ and let $S$ be the spike mechanism. Then if the approximation ratio of $S$ given $\mathbf{x}'$ is not higher than 2, then so is that of $S$ given $\mathbf{x}$: 	
					$$\frac{\SC(S,\bx')}{\SC(\OPT,\bx')} \leq 2 \Rightarrow  \frac{\SC(S,\bx)}{\SC(\OPT,\bx)} \leq 2$$
				\end{lemma}

				Since the cumulative function defining the spike mechanism is symmetrical, the claim can be trivially extended to right-compressions as well.
				
				After reapplying compressions on both sides, the resulting profile has agents in three locations at most (see Figure \ref{fig:spike-4}). The last lemma in the proof states that in this final stage, the ratio is not higher than 2:
				
				\begin{lemma} \label{lemma:three-candidates}
					Let $\bx$ be a tight location profile in which there are at most 3 active candidates: $y_{opt-1}<y_{opt}<y_{opt+1}$. The ratio of the $S$ given $\bx$ is not higher than 2: 
					$ \frac{\SC(S, \bx)}{\SC(\OPT,\bx)} \leq 2$.
				\end{lemma}	
				
			\end{proof}

%% file: randomized_sv.tex
\section{Additional Results for Randomized Mechanisms} \label{sec-randomized}

	\subsection{Lower Bounds}
	In this section show lower bounds of randomized mechanisms in different settings.
	When the network is $\Re$, we present a lower bound of 2 for any truthful in expectation mechanism, even if it is a location mechanism. 
	By the hierarchy presented in Theorem \ref{thm:classes}, this lower bound trivially holds for truthful in expectation ranking and voting mechanisms as well (see Figure \ref{fig:classes}).
	Ergo, spike is optimal over all truthful in expectation mechanisms.
	
	Additionally, we show a lower bound of 2 for any randomized ranking mechanism, even when the mechanism need not be truthful (in the non-strategic setting). 
	These results prove that the approximation ratio achieved by spike is tight. 
	For more general metric spaces the lower bound changes ---
	In $\Re^d$, we show a lower bound for any truthful voting mechanism of $3-\frac{2}{d+1}$. 
	We also present a lower bound of $7/3$ for any truthful ranking mechanism in $\Re^2$ (this bound also holds for $\Re^d$, for any $d>2$).
	
	We start by proving a helpful lemma. Informally, the lemma states that when there is an agent located on a border (and can therefore submit several different truthful actions to a ranking or voting mechanism), her cost should not change under any truthful report she submits.

	\begin{lemma} \label{lemma:border-equal}
		For any truthful in expectation ranking mechanism $M$ in any metric space, let $b_{i,j}$ be the border between ranking zones $\calr_i, \calr_j$.
		Let $\pi_i,\pi_j$ be the rankings consistent with $\calr_i,\calr_j$ respectively.
		Let agent $l$ be located on this border, that is: $x_l \in b_{i,j}$.
		\newline
		Then the cost at point $x_l$ remains the same whether the agent reports $\pi_i$ or $\pi_j$, that is:
		$$\cost_{x_l}(M,(a_l=\pi_i,a_{-l})) = \cost_{x_l}(M,(a_l=\pi_j,a_{-l}))$$
	\end{lemma} 
	
	The proof is given in appendix \ref{prf:border-equal}.
	
	\begin{observation} \label{obs:border-equal-voting}
		The previous lemma also holds for voting mechanisms.
	\end{observation}
	The proof of the observation follows the exact same lines as the proof of the lemma.
	
		\begin{theorem} \label{thm:simplex}
			In the $d$ dimensional real space $\Re^d$, any truthful in expectation voting mechanism has an approximation ratio of at least $3-\frac{2}{d+1}$.
		\end{theorem}
		\begin{proof}
			Let there be $d+1$ candidates, located on the vertices of a regular simplex $H$ (all $d+1$ vertices are equally distanced from one another). Let there be $d+1$ agents, and let $M$ be an arbitrary truthful in expectation voting mechanism.
			
			Let $\bx$ be the profile in which each agent $i$ is located precisely on candidate $\calc_i$. 
			Therefore $\mathbf{a}=(\calc_1,\calc_2 \ldots \calc_{d+1})$ is the only truthful voting profile for $\bx$.
			Denote the probability of choosing candidate $i$ as $p_{i}(a)$, that is: $p_{i}(a)=\Pr(M(\mathbf{a})=\calc_{i})$.
			Clearly there exists some candidate which is chosen by $M$ with probability at least $\frac{1}{d+1}$. 
			Assume without loss of generality that this candidate is $\calc_{d+1}$, that is: $p_{d+1}(a) \geq \frac{1}{d+1}$. 		
			
			We move on to define another location profile, $\bx'$, which is also consistent with the voting profile $\mathbf{a}$.
			Let $H'$ be the regular simplex in which candidates $\calc \setminus \calc_{d+1}$ are on the vertices. 
			Let $P$ be the point with equal distance to all $d$ vertices in $H'$ ($H'$ is a regular simplex, so such a point necessarily exists). 
			Denote this distance as $t$. 
			However, this distance is different from the distance from $P$ to $\calc_{d+1}$: $|P-y_{d+1}| = u \neq t$.
			Let $\bx'$ be the profile in which there are $k$ agents at $P$ and one agent at $\calc_{d+1}$ (see Figure \ref{fig:equilateral-triangle-voting}). 
			
			According to Observation \ref{obs:border-equal-voting}, the cost of an agent at point $P$ should not change under any truthful vote, that is for any vote $\calc_j: 1 \leq j \leq d$.
			In particular, this holds when any agent on point $P$ votes for candidate $\calc_1$.
			We make use of this observation several times by changing the votes for each of the points at $P$ to $\calc_1$, one at a time, such that the final voting profile is $\mathbf{a'}=\{\calc_1,\calc_1 \ldots \calc_1,\calc_{d+1}\}$ ($d$ agents vote for $\calc_1$, one agent votes for $\calc_{d+1}$)..
			Due to the observation, the cost of point $P$ must remain the same throughout these transitions, that is: $\cost_P(M,a')=\cost_P(M,a)$. 
			
			Therefore: 
			\begin{eqnarray*}
				&& \cost_P(M,a) = \cost_P(M,a') \\ 
				&\Rightarrow& u \cdot p_{d+1}(a) + t \cdot (1-p_{d+1}(a)) = u \cdot p_{d+1}(a') + t \cdot (1-p_{d+1}(a')) \\
				&\Rightarrow& t + (u-t)\cdot p_{d+1}(a) = t + (u-t)\cdot p_{d+1}(a') \\
				&\Rightarrow& p_{d+1}(a') = p_{d+1}(a) \\
				&\Rightarrow& p_{d+1}(a') \geq \frac{1}{d+1}
			\end{eqnarray*}				
			
			Denote the midpoint between $y_1$ and $y_{d+1}$ as $Q$.
			Without loss of generality, scale the distances such that $|y_1-Q|=|Q-y_{d+1}|=1$.
			Examine the following location profile $\bx''= (\calc_1,\calc_1 \ldots \calc_1, Q)$, which is also consistent with the voting profile $\mathbf{a'}$.
			In this case the cost of $\calc_1$, which is the optimal candidate, is: $\SC(\calc_1,\bx'')=1$.
			The cost of $\calc_{d+1}$ is $\SC(\calc_{d+1},\bx'') =d\cdot 2 + 1 = 2d+1$.
			Therefore the approximation ratio of $M$ is at least:
			\begin{eqnarray*}
				\frac{\SC(M,\bx'')}{\SC(\OPT,\bx'')} &=& p_{d+1}(a') (2d+1) + (1-p_{d+1}(a'))(1) \\
				&=& 2d \cdot p_{d+1}(a') + 1 \\
				&\geq& (2d) \frac{1}{d+1}  + 1 \\
				&=&	\frac{2d+2-2}{d+1}+1 = 3 - \frac{2}{d+1}
			\end{eqnarray*}

			\begin{figure}[h]
				
		\setlength{\fboxrule}{0.5pt}
		\noindent \fbox{\noindent\makebox[1\textwidth][c]
			{\begin{minipage}{1\textwidth}	
				
				\centering		
				
				\begin{subfigure}{.33\textwidth}

					\begin{tikzpicture}[y=.3cm, x=.3cm,font=\sffamily]			
					
					\draw[-] (0,0) -- (7.5,0) node{};
					\draw[-] (0,0) -- (3.75, 6.5) node{};
					\draw[-] (7.5,0) -- (3.75, 6.5) node{};
					\draw[-] (3.75,2.165) -- (3.75,0) node[midway,right]{};	
					\draw[-] (3.75,2.165) -- (1.875,3.25) node[midway,left]{};				
					\draw[-] (3.75,2.165) -- (5.625,3.25) node[midway,right,xshift=0.1cm]{};	 
					
					\draw[fill=white] (0, 0) circle (0.5) node [left, yshift=0.3cm] {$y_1,x_1$};
					\draw[fill=white] (3.75, 6.5) circle (0.5) node [above, yshift=0.3cm] {$y_2,x_2$};
					\draw[fill=white] (7.5, 0) circle (0.5) node [right, yshift=0.3cm] {$y_3,x_3$};
					
					\draw[fill=black] (3.75,2.165) circle (0) node [right, yshift=0.1cm] {};	
					\draw[fill=black] (1.875,3.25) circle (0) node  [left,yshift=0.1cm] {$P$};  
					\draw[fill=black] (3.75,0) circle (0) node  [below,yshift=-0.1cm] {$Q$}; 
					
					\draw[fill=black] (0,0) circle (0.2) node [below,yshift=-0.2cm] {};	
					\draw[fill=black] (3.75, 6.5) circle (0.2) node  [right,yshift=0.2cm] {};  
					\draw[fill=black] (7.5, 0) circle (0.2) node  [below,yshift=-0.1cm] {};
					
					
					\end{tikzpicture}					
				\end{subfigure}%
				\begin{subfigure}{.33\textwidth}

					\begin{tikzpicture}[y=.3cm, x=.3cm,font=\sffamily]			
					
					\draw[-] (0,0) -- (7.5,0) node{};
					\draw[-] (0,0) -- (3.75, 6.5) node{};
					\draw[-] (7.5,0) -- (3.75, 6.5) node{};
					\draw[-] (3.75,2.165) -- (3.75,0) node[midway,right]{};	
					\draw[-] (3.75,2.165) -- (1.875,3.25) node[midway,left]{};				
					\draw[-] (3.75,2.165) -- (5.625,3.25) node[midway,right,xshift=0.1cm]{};	 
					
					\draw[fill=white] (0, 0) circle (0.5) node [left, yshift=0.3cm] {$y_1$};
					\draw[fill=white] (3.75, 6.5) circle (0.5) node [above, yshift=0.3cm] {$y_2$};
					\draw[fill=white] (7.5, 0) circle (0.5) node [right, yshift=0.3cm] {$y_3,x_3$};
					
					\draw[fill=black] (3.75,2.165) circle (0) node [right, yshift=0.1cm] {};	
					\draw[fill=black] (3.75,0) circle (0) node  [below,yshift=-0.1cm] {$Q$}; 
					
					\draw[fill=black] (1.875,3.25) circle (0.2) node [left] {$P,x_1,x_2$};	
					\draw[fill=black] (7.5, 0) circle (0.2) node  [below,yshift=-0.1cm] {};
					
					
					\end{tikzpicture}				
				
				\end{subfigure}%
				\begin{subfigure}{.33\textwidth}

					\begin{tikzpicture}[y=.3cm, x=.3cm,font=\sffamily]			
					
					\draw[-] (0,0) -- (7.5,0) node{};
					\draw[-] (0,0) -- (3.75, 6.5) node{};
					\draw[-] (7.5,0) -- (3.75, 6.5) node{};
					\draw[-] (3.75,2.165) -- (3.75,0) node[midway,right]{};	
					\draw[-] (3.75,2.165) -- (1.875,3.25) node[midway,left]{};				
					\draw[-] (3.75,2.165) -- (5.625,3.25) node[midway,right,xshift=0.1cm]{};	 
					
					\draw[fill=white] (0, 0) circle (0.5) node [left, yshift=0.3cm] {$y_1,x_1,x_2$};
					\draw[fill=white] (3.75, 6.5) circle (0.5) node [above, yshift=0.3cm] {$y_2$};
					\draw[fill=white] (7.5, 0) circle (0.5) node [right, yshift=0.3cm] {$y_3$};
					
					\draw[fill=black] (3.75,2.165) circle (0) node [right, yshift=0.1cm] {};	
					\draw[fill=black] (1.875,3.25) circle (0) node  [left,yshift=0.1cm] {$P$};  
					
					\draw[fill=black] (0,0) circle (0.2) node [below,yshift=-0.1cm] {};	
					\draw[fill=black] (3.75, 0) circle (0.2) node  [below,yshift=-0.1cm] {$Q,x_3$};
					
					\draw (0, -2.7) -- (0, -3.3) node [above]{};	 				
					\draw[dim] (0,-3) -- (3.75,-3) node[midway,below] {$1$};	
					\draw (3.75, -2.7) -- (3.75, -3.3) node [above]{};	 								
					
					\end{tikzpicture}				
				
				\end{subfigure}

%
%
%
%
%
				\caption{An illustration of the proof of Theorem \ref{thm:simplex} for the case of $d=2$. 
					\newline
					The three candidates are on the vertices of an equilateral triangle (a regular simplex with 3 vertices). 
					The lines within the triangle denote the borders between a pair of candidates. 
					The simplex $H'$ is the line between $y_1,y_2$, and its midpoint is $P$. 
					\newline
					These three figures, from left to right, show the dynamics of the proof (location profiles $\bx$,$\bx'$ and $\bx''$ respectively). The key observation is that the agents at point $P$ are at equal distance to all candidates except $\calc_3$, therefore the transition does not change the probability of $\calc_3$ to be chosen.
				}

				\label{fig:equilateral-triangle-voting}

				\end{minipage}}} \par\setlength{\fboxrule}{0.2pt}						
			\end{figure}
			
		\end{proof}	
		
		\begin{observation} \label{obs:lb-location-2}
			Any truthful in expectation location mechanism has an approximation ratio of at least 2, even on the line.
		\end{observation}				
		\begin{proof}
			Let there be two candidates on the line. According to Theorem \ref{thm:classes}, any truthful in expectation location mechanism is equivalent to a voting mechanism. One can apply the same proof as in Theorem \ref{thm:simplex} for the case of $d=1$ (in this case, both point $P$ and point $Q$ are the midpoint between the two candidates), to achieve a lower bound of $3-\frac{2}{d+1}=2$.			
		\end{proof}

		\begin{theorem} \label{thm:triangle-ranking}
			In $\Re^2$, any truthful in expectation ranking mechanism has an approximation ratio of at least $7/3$.
		\end{theorem}
		\begin{proof}
			Let there be 3 candidates located such that they form an equilateral triangle, and let $M$ be a truthful in expectation ranking mechanism. Let $\bara=(a_1,a_2,a_3)$ be the following ranking profile:
			\begin{eqnarray*}
				a_1 = \calc_1 \succeq \calc_2 \succeq \calc_3 \\
				a_2 = \calc_2 \succeq \calc_3 \succeq \calc_1 \\
				a_3 = \calc_3 \succeq \calc_1 \succeq \calc_2 
			\end{eqnarray*}		
			Let $\bx$ be some location profile consistent with $\bara$ (see Figure \ref{fig:equilateral-triangle-ranking}).
			Denote $p_i(\bara) = \Pr[M(\bara)=\calc_i]$. 
			From symmetry, there exists some candidate chosen with probability at least $1/3$. 
			Assume without loss of generality that this is candidate $\calc_3$, that is: $p_3(\bara) \geq 1/3$.
			
			Let $P_1$ be a point such that $|P_1-y_1|=|P_1-y_2|=t_1$, and $|P_1-y_3|=u_1$, where $t_1 \neq u_1$.
			Let $\bx'=(P_1,x_2,x_3)$. 
			Let $a_1' = \calc_2 \succeq \calc_1 \succeq \calc_3$ and let $\bara' = (a_1',a_2,a_3)$.
			Notice that $\bx'$ is consistent with both $\bara$ and $\bara'$, therefore according to Lemma \ref{lemma:border-equal} the cost at $P_1$ should remain the same for $\bara,\bara'$:
			\begin{eqnarray*}
				&& \cost_{P_1}(M,\bara) = \cost_{P_1}(M,\bara') \\ 
				&\Rightarrow& u_1 \cdot p_{3}(\bara) + t_1 \cdot (1-p_{3}(\bara)) = u_1 \cdot p_{3}(\bara') + t_1 \cdot (1-p_{3}(\bara'))  \\
				&\Rightarrow& t_1 + (u_1-t_1)\cdot p_{3}(\bara) = t_1 + (u_1-t_1)\cdot p_{3}(\bara') \\
				&\Rightarrow& p_{3}(\bara') = p_{3}(\bara) \\
				&\Rightarrow& p_{3}(\bara') \geq \frac{1}{3}
			\end{eqnarray*}	
			
			Let $P_2$ be a point such that $|P_2-y_2|=|P_2-y_1|=t_2$, and $|P_2-y_3|=u_2$, where $t_2 \neq u_2$. 
			Let $\bx''=(x_1',x_2,P_2)$.
			Let $\bara_3'' = \calc_3 \succeq \calc_2 \succeq \calc_1$ and let $\bara'' = (a_1',a_2,a_3'')$.
			According to Lemma \ref{lemma:border-equal} the cost at $P_2$ should remain the same for $\bara',\bara''$:
			\begin{eqnarray*}
				\cost_{P_2}(M,\bara') &=& \cost_{P_2}(M,\bara'') \\
				&\Rightarrow&  u_2 \cdot p_{3}(\bara) + t_2 \cdot (1-p_{3}(\bara')) = u_2 \cdot p_{3}(\bara'') + t_2 \cdot (1-p_{3}(\bara'')) \\
				&\Rightarrow& p_{3}(\bara'') = p_{3}(\bara') \\
				&\Rightarrow& p_{3}(\bara'') \geq \frac{1}{3}
			\end{eqnarray*}		
			
			Let $Q$ be the midpoint between $y_2,y_3$, and let $\bx'''=(y_2,y_2,Q)$. 
			Without loss of generality, scale the distances such that $|y_3-Q|=|Q-y_2|=1$. 
			Therefore the cost of $\calc_{2}$, the optimal candidate, is: $\SC(\calc_2,\bx''',\bara'')=\SC(\OPT,\bx''',\bara'')=1$. 
			The cost of $\calc_{3}$ is: $\SC(\calc_3,\bx''',\bara'') = 2\cdot 2 + 1 = 5$.
			Therefore the approximation ratio of $M$ is at least:
			\begin{eqnarray*}
				\frac{\SC(M,\bx''',\bara'')}{\SC(\OPT,\bx''',\bara'')} = p_{3}(\bara'')\cdot 5 + (1-p_{3}(\bara''))\cdot 1 &=& 1 + 4\cdot p_{3}(\bara'') \geq 1 + \frac{4}{3} = \frac{7}{3}
			\end{eqnarray*}		
		
			\begin{figure}[h]
				
		\setlength{\fboxrule}{0.5pt}
		\noindent \fbox{\noindent\makebox[1\textwidth][c]
			{\begin{minipage}{1\textwidth}	
				
				\centering		
				
				\begin{subfigure}{.25\textwidth}

					\begin{tikzpicture}[y=.3cm, x=.3cm,font=\sffamily]			
					
					\draw[-] (0,0) -- (7.5,0) node{};
					\draw[-] (0,0) -- (3.75, 6.5) node{};
					\draw[-] (7.5,0) -- (3.75, 6.5) node{};
					\draw[-] (3.75,6.45) -- (3.75,0) node[midway,right]{};	
					\draw[-] (7.5,0) -- (1.875,3.25) node[midway,left]{};				
					\draw[-] (0,0) -- (5.625,3.25) node[midway,right,xshift=0.1cm]{};	 
					
					\draw[fill=white] (0, 0) circle (0.5) node [below,yshift=-0.1cm] {$y_1$};
					\draw[fill=white] (3.75, 6.5) circle (0.5) node [above, yshift=0.1cm] {$y_2$};
					\draw[fill=white] (7.5, 0) circle (0.5) node [below,yshift=-0.1cm] {$y_3$};
					
					
					\draw[fill=black] (2,2) circle (0.2) node [left] {$x_1$};	
					\draw[fill=black] (4.5, 4) circle (0.2) node  [right] {$x_2$};  
					\draw[fill=black] (5, 0.75) circle (0.2) node  [right] {$x_3$};
					
					
					\end{tikzpicture}					
				\end{subfigure}%
				\begin{subfigure}{.25\textwidth}

					\begin{tikzpicture}[y=.3cm, x=.3cm,font=\sffamily]			
					
					\draw[-] (0,0) -- (7.5,0) node{};
					\draw[-] (0,0) -- (3.75, 6.5) node{};
					\draw[-] (7.5,0) -- (3.75, 6.5) node{};
					\draw[-] (3.75,6.45) -- (3.75,0) node[midway,right]{};	
					\draw[-] (7.5,0) -- (1.875,3.25) node[midway,left]{};				
					\draw[-] (0,0) -- (5.625,3.25) node[midway,right,xshift=0.1cm]{};	 
					
					\draw[fill=white] (0, 0) circle (0.5) node [below,yshift=-0.1cm] {$y_1$};
					\draw[fill=white] (3.75, 6.5) circle (0.5) node [above, yshift=0.1cm] {$y_2$};
					\draw[fill=white] (7.5, 0) circle (0.5) node [below,yshift=-0.1cm] {$y_3$};
					
					
					\draw[fill=black] (2.5,2.88675) circle (0.2) node [left] {$P_1,x_1$};	
					\draw[fill=black] (4.5, 4) circle (0.2) node  [right] {$x_2$};  
					\draw[fill=black] (5, 0.75) circle (0.2) node  [right] {$x_3$};
					
					
					\end{tikzpicture}				
				
				\end{subfigure}%
				\begin{subfigure}{.25\textwidth}

					\begin{tikzpicture}[y=.3cm, x=.3cm,font=\sffamily]

					\draw[-] (0,0) -- (7.5,0) node{};
					\draw[-] (0,0) -- (3.75, 6.5) node{};
					\draw[-] (7.5,0) -- (3.75, 6.5) node{};
					\draw[-] (3.75,6.45) -- (3.75,0) node[midway,right]{};	
					\draw[-] (7.5,0) -- (1.875,3.25) node[midway,left]{};				
					\draw[-] (0,0) -- (5.625,3.25) node[midway,right,xshift=0.1cm]{};	 
					
					\draw[fill=white] (0, 0) circle (0.5) node [below,yshift=-0.1cm] {$y_1$};
					\draw[fill=white] (3.75, 6.5) circle (0.5) node [above, yshift=0.1cm] {$y_2$};
					\draw[fill=white] (7.5, 0) circle (0.5) node [below,yshift=-0.1cm] {$y_3$};
					
					
					\draw[fill=black] (2.5,2.88675) circle (0.2) node [left] {$P_1,x_1$};	
					\draw[fill=black] (4.5, 4) circle (0.2) node  [right] {$x_2$};  
					\draw[fill=black] (6.25,0.7225) circle (0.2) node  [right,yshift=0.1cm] {$P_2,x_3$};
					
					
					\end{tikzpicture}				
				
				\end{subfigure}%
				\begin{subfigure}{.25\textwidth}

					\begin{tikzpicture}[y=.3cm, x=.3cm,font=\sffamily]

					\draw[-] (0,0) -- (7.5,0) node{};
					\draw[-] (0,0) -- (3.75, 6.5) node{};
					\draw[-] (7.5,0) -- (3.75, 6.5) node{};
					\draw[-] (3.75,6.45) -- (3.75,0) node[midway,right]{};	
					\draw[-] (7.5,0) -- (1.875,3.25) node[midway,left]{};				
					\draw[-] (0,0) -- (5.625,3.25) node[midway,right,xshift=0.1cm]{};	 
					
					\draw[fill=white] (0, 0) circle (0.5) node [below,yshift=-0.1cm] {$y_1$};
					\draw[fill=white] (3.75, 6.5) circle (0.5) node [above, yshift=0.1cm] {$y_2,x_1,x_2$};
					\draw[fill=white] (7.5, 0) circle (0.5) node [below,yshift=-0.1cm] {$y_3$};
					
					
					\draw[fill=black] (3.75, 6.5) circle (0.2) node [left] {};	
					\draw[fill=black] (5.625, 3.25) circle (0.2) node  [right] {$Q,x_3$};  
					
					\draw[dim] (6.625,3.872) -- (4.75,7.122) node[midway,right] {$1$};	
					
					\end{tikzpicture}				
				
				\end{subfigure}							
			
				\caption{An illustration of the proof of Theorem \ref{thm:triangle-ranking}. \newline
					These four figures, from left to right, show the dynamics of the proof (profiles $\bx$, $\bx'$, $\bx''$ and $\bx'''$ respectively). 
				}				
				\label{fig:equilateral-triangle-ranking}

				\end{minipage}}} \par\setlength{\fboxrule}{0.2pt}						
			\end{figure}
			
		\end{proof}	
	\begin{lemma} \label{lemma:rand-lb-2-non-strategic}
		No randomized ranking mechanism can achieve an approximation ratio strictly below 2, even if the metric is $\Re$ and even if there are no truthfulness requirements from the mechanism (non-strategic case).
	\end{lemma}
	The proof is deferred to appendix \ref{prf:rand-lb-2-non-strategic}.

	
	\subsection{Upper Bound}

	We previously showed that spike achieves an approximation ratio of 2 on the line. We now show that for a general metric space, random dictator achieves a ratio of 3.
		
	Random dictator is a simple randomized mechanism which achieves an approximation ratio of 2 in the continuous model. Recall that random dictator locates the facility on vote $a_i$ with probability $1/n$ for all $i \in N$. We show that in our model, it yields a competitive ratio of 3 for any metric space (but it cannot improve beyond 3 even on the line). 
	
	\begin{lemma} \label{lemma:rd-3}
		On any metric space, random dictator achieves an approximation ratio of exactly $3$.
	\end{lemma}
	The proof can be found in appendix \ref{prf:rd-3}.

	As opposed to the continuous model, in our candidate model random dictator is not group-strategyproof. Refer to \ref{GSP} in the appendix for the definition of group-strategyproofness, and for proof of this statement.

%% file: deterministic_sv.tex
\section{Deterministic Mechanisms} \label{sec-deterministic}
	\subsection{Lower Bound}

	In the continuous model it is well known that choosing the location of the median agent is both truthful and optimal \cite{procaccia2009approximate}. The following lemma shows that this does not hold in our candidate model.

	\begin{lemma} \label{lemma:lb-det-3}
		No deterministic truthful mechanism (location, ranking or voting) can achieve an approximation ratio strictly below $3$ for the social cost, even when the metric space is $\Re$.
	\end{lemma}
	\begin{proof}
		Let there be 2 candidates such that $y_1 = -1$ and $y_2 = 1$. According to the sixth claim in Theorem \ref{thm:classes} (Claim \ref{claim:two-candidates}), any truthful location mechanism $M$ is necessarily reducible to a voting mechanism. Let $\mathbf{x}=(-1,\epsilon)$, $\mathbf{x}'=(-\epsilon,1)$ be two location profiles, and let $B$ be the border between them. Clearly, both profiles correspond with the same votes ($x_1,x_1' \in \calv_1 \setminus B$ and $x_2,x_2' \in \calv_2 \setminus B$), therefore their outcome will be identical.
		
		If $M(\mathbf{x})=M(\mathbf{x}')=\calc_1$, then the approximation ratio for $M$ given $\bx$ is $\frac{\SC(M,\mathbf{x})}{\SC(\OPT,\mathbf{x})}=\frac{3-\epsilon}{1+\epsilon}$. 
		
		If $M(\mathbf{x})=M(\mathbf{x}')=\calc_2$ then the approximation ratio for $M$ given $\mathbf{x}'$ is $\frac{\SC(M,\mathbf{x}')}{\SC(\OPT,\mathbf{x}')}=\frac{3-\epsilon}{1+\epsilon}$. 
		
		In either case, the approximation ratio tends to $3$ as $\epsilon$ tends to $0$.
	\end{proof}
	
	A result by Anshelevich et al. (Theorem 3 in \cite{anshelevich2014approximating}) uses a similar method to show the worst case for ranking mechanisms with non-strategic agents is 3 (the proof in that paper was done for a general metric space, but it can also be applied to $\Re$). Lemma \ref{lemma:lb-det-3} makes use of Claim \ref{claim:two-candidates} to prove that the same lower bound of 3 also applies even to location mechanisms, under the condition that the mechanism must be truthful (when the agents are strategic). 
	
	\subsection{Upper Bound}
	
	As noted previously, in the continuous model on $\Re$, the mechanism which locates the facility on the report of the median agent is truthful and results in the optimal social cost. We define the median mechanism in the context of candidate constraints, and assess its social cost. 
	
	\begin{definition} [Median mechanism] 
		Given a vote profile $\bara$, let $\pi$ be a permutation such that $y\left( a_{\pi(1)} \right) \leq \ldots \leq y\left( a_{\pi(n)} \right)$. 
		Median is a voting mechanism which chooses the median vote, that is $a_{\pi(\ceil*{n/2})}$.
	\end{definition}

	\begin{lemma} \label{lemma:det-med-3}
		Median is a truthful mechanism which results in a 3 approximation of the optimal social cost on $\Re$.
	\end{lemma}	
	The proof is given in the appendix \ref{prf:det-med-3}.

%% file: discussion_sv.tex
\section{Open Problems}

As explained in the introduction, we believe that there are many possible manifestations of this setting, though it has barely been investigated, so there is plenty of room for future work. 
Clearly, there are gaps in some of the bounds in the paper which are currently open.
Additionally, one can investigate the questions from the facility location literature in this setting, e.g., electing a committee of several candidates, addressing the problem of ``obnoxious facility location" in which the agents wish to elect the candidate farthest away from them, or proving additional bounds for group strategyproofness.

%% file: appendix.tex
\section{Appendix}
	
	\subsection{Missing Proofs from Section \ref{sec-classes}}
	
	\label{def:reduciblity}
	
	In this part we aim to define reducibility of some mechanism type (voting, ranking or location) to some other mechanism type. The definition of the reduction requires a couple of additional definitions, which we will express formally as well as explain intuitively.
	
	Intuitively, we say that a mechanism type $\hat{A}$ is of finer granularity than mechanism type $\hat{B}$, if the information of a true action in $\hat{A}$ can determine a true action in $\hat{B}$. 
	For instance, a location determines a ranking (or several rankings, if a point is on a border), therefore location mechanisms are of finer granularity than ranking mechanisms. 
	Similarly, ranking mechanisms are of finer granularity than voting mechanisms, and location mechanisms are also of finer granularity than voting mechanisms.
	However, ranking mechanisms are not of finer granularity than location mechanisms, since a ranking does not determine a location. That is, for a ranking $\pi_1$ there exist different locations $x_1, x_2$ whose true ranking is $\pi_1$.
	Formally, a mechanism type $\hat{A}: A \rightarrow \calc$ is of {\em finer granularity} than mechanism type $\hat{B}:B\rightarrow \calc$ if for any point $x$ and any $a \in A$ there exists some $b \in B$ such that: if $a$ is a true action of an agent at point $x$ under $\hat{A}$, then $b$ is a true action of point $x$ under $\hat{B}$. 
	In this case we denote $\hat{A} \succ \hat{B}$.
	
	We now utilize this notion of granularity to define consistent functions. Intuitively, we would like to define functions which map between inputs of different mechanism types in a ``consistent" manner. 
	For instance, when mapping from rankings (the input to ranking mechanisms) to votes (candidates, that is - the input to voting mechanisms), we search for functions which map each ranking to the top candidate in that ranking.
	When mapping from votes to rankings, we seek functions which map a vote for a candidate to some ranking in which this candidate is first.
	Formally, given mechanism types $\hat{A}: A \rightarrow \calc$ and $\hat{B}:B\rightarrow \calc$ such that $\hat{A} \succ \hat{B}$, a function $f: A \rightarrow B$ is called {\em consistent} if for any point $x$ and for any $a \in A$, then if $a$ is a true action under $\hat{A}$, then $f(a)$ is a true action under $\hat{B}$. 
	Notice that in these cases a function $f$ is unique (except for the cases in which $\hat{A}$ is a location mechanism and $x$ is a point on a border). 
	If $\hat{B} \succ \hat{A}$ then we define $f:A \rightarrow B$ to be consistent if for any point $x$, if $f(a)$ is a true action for $x$ under $\hat{B}$ then $a$ is a true action for $x$ under $\hat{A}$. In these cases the function $f$ is not unique (for example, there are several rankings in which a specific candidate is first). 
	
	The function $f$ may be randomized, as long as it is a randomization over deterministic consistent functions. For example, a consistent function $f$ mapping locations (the input of location mechanisms) to candidates (the inputs of voting mechanisms) must map every point which is not on a border to their favorite candidate. On the other hand, for a point $x$ on some border, $f$ may randomize the output of $x$ arbitrarily over the set of favorite candidates of $x$.

	A candidate selection mechanism $M$ (whether location, ranking or voting) is said to be {\em reducible} to a candidate selection mechanism $M'$ (location, ranking, or voting) if there exists a consistent function $f$ mapping every action profile $\bara$, which is an input of $M$, to some action profile $f(\bara)=\bara'$ which is the input of $M'$, such that the distribution over the candidates, $M(\bara)$, is identical to the distribution over candidates $M'(\bara')$ (see Figure \ref{fig:reduction}). 

		\begin{figure}[t]
			
			\setlength{\fboxrule}{0.5 pt}
			\noindent \fbox{\noindent\makebox[1\textwidth][c]
				{\begin{minipage}{1\textwidth}
						
						\centering
	
						\includegraphics[scale=0.7]{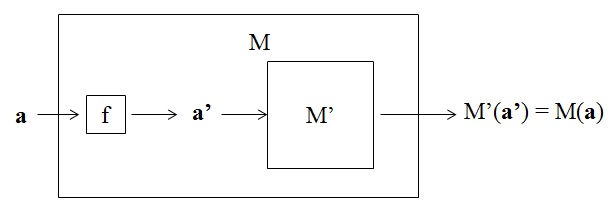}
						
							\caption{A graphic sketch of a mechanism $M$ which is reducible to a mechanism $M'$. }					
						\label{fig:reduction}
					
				\end{minipage}}} \par\setlength{\fboxrule}{0.2pt}

		\end{figure}

	For example, every voting mechanism, $M$, is reducible to some ranking mechanism $M'$. 
	The reduction is as follows: Let $f$ be the consistent function describes previously -- it receives a vector of $n$ candidates, and outputs a vector of $n$ rankings (permutations), where for each $i$, the $i$'th candidate is ranked first in the $i$'th ranking.
	We choose a ranking mechanism $M'$ which ignores all entries in the rankings but the first, and simulates the voting mechanism $M$ on the top entries of the rankings. 
	By definition, $M$ is reducible to $M'$. 
	Moreover, note that $M$ was an arbitrary voting mechanism, so we conclude that indeed every voting mechanism is reducible to some ranking mechanism. 
	This logic further shows that for mechanism types $\hat{A},\hat{B}$ such that $\hat{B} \succ \hat{A}$, any mechanism $M$ of type $\hat{A}$ is reducible to some mechanism $M'$ of type $\hat{B}$.
	
	As written in Section \ref{sec-classes}, two sets of mechanisms, $S_1$ and $S_2$, are said to be {\em equivalent} if every $M_1 \in S_1$ is reducible to some $M_2 \in S_2$, and every $M_2 \in S_2$ is reducible to some $M_1 \in S_1$. 	
	A set of mechanisms $S_1$ is said to be {\em strictly contained} in a set of mechanisms $S_2$ if every mechanism $M_1 \in S_1$ is reducible to some mechanism $M_2 \in S_2$, yet not every mechanism $M_2\in S_2$ is reducible to some mechanism $M_1\in S_1$. 
	This is a slight abuse of terminology since the sets $S_1$ and $S_2$ may be disjoint, as their input space may be different.

	The following lemma will be of use in the main theorem of this section.
	
	\begin{lemma} \label{lemma:reducible-granular}
		In any metric space, let $\hat{A},\hat{B}$ be mechanism types such that $\hat{B} \succ \hat{A}$. Let $S_1,S_2$ be the sets of truthful mechanisms of type $\hat{A},\hat{B}$ respectively. Then for any $M_1 \in S_1$ there exists some $M_2 \in S_2$ such that $M_1$ is reducible to $M_2$.
	\end{lemma}
	\begin{proof}
		By the fact that $\hat{B}$ is of finer granularity than $\hat{A}$ and by the definition of the consistent function $f$, it is eminent that $M_1$ is reducible to some $M_2$ of type $\hat{B}$ since $M_2$ can completely disregard any input beyond any information in $\hat{A}$ and simulate $M_1$ (as explained previously for the example in which a ranking mechanism disregards any candidate except for the top candidate in each ranking).
		
		In order to complete the proof, it is left to show that there exists such a mechanism $M_2$ which is truthful.
		Since the only reports which change the outcome of $M_2$ are consistent with reports which would change the outcome of $M_1$, then if $M_2$ weren't truthful this would contradict the truthfulness of $M_1$.
	\end{proof}
	\begin{observation} \label{obs:reducible-granular}
		Notice that this reasoning also holds for truthful in expectation mechanisms (that is if $S_1,S_2$ are defined as the sets of truthful in expectation mechanisms of types $\hat{A},\hat{B}$).
	\end{observation}

	We move on to proving the main theorem of this section:		
	\begin{proof} of Theorem \ref{thm:classes}: \label{prf:classes}
		The proof of each claim is given separately:

		\begin{claim} \label{lemma:det-ranking-voting}
			The class of truthful deterministic ranking mechanisms strictly contains the class of truthful deterministic voting mechanisms.
		\end{claim}
		\begin{proof}
			According to Lemma \ref{lemma:reducible-granular} any truthful voting mechanism is reducible to some truthful ranking mechanism. 
			We exhibit a deterministic truthful ranking mechanism $M$ which is not reducible to any voting mechanism (even on $\Re$). 
			We show this by exhibiting an example in which $M$ acts differently under two ranking profiles which are mapped by any consistent function to the same voting zone. 
			
			Let there be 3 candidates, and denote the ranking zones $\calr_1,\calr_2, \calr_3, \calr_4$, which match the permutations over candidates $\pi_1,\pi_2,\pi_3,\pi_4$ respectively (see Figure \ref{fig:4-subzones}). 
			$M$ acts as follows: 
			If $a_1 \in \pi_1 \cup \pi_2$ and $a_2 \in \pi_1 \cup \pi_2$, then choose $\calc_1$. 
			Otherwise, choose $\calc_3$.
			
			\begin{figure}[t]
				
				\setlength{\fboxrule}{0.5 pt}
				\noindent \fbox{\noindent\makebox[1\textwidth][c]
					{\begin{minipage}{1\textwidth}
							
							\centering		

							\begin{tikzpicture}[y=.3cm, x=.3cm,font=\sffamily]			
							\draw[<->, thick] (0,0) -- (40,0) node{};
							\draw[fill=white] (5, 0) circle (0.8) node [above, yshift=0.3cm] {$\calc_1$};
							\draw[fill=white] (15, 0) circle (0.8) node [above, yshift=0.3cm] {$\calc_{2}$};
							\draw[fill=white] (35, 0) circle (0.8) node [above, yshift=0.3cm] {$\calc_{3}$};
							
							\draw (10, 1) -- (10, -1) node [above, , yshift=0.5cm, align=center]{$b_{1,2}$};
							\draw (20, 1) -- (20, -1) node [above, , yshift=0.5cm, align=center]{$b_{1,3}$};		
							\draw (25, 1) -- (25, -1) node [above, yshift=0.5cm, align=center]{$b_{2,3}$};
							
							\draw[dim] (0,-3) -- (10,-3) node[midway,above] {$\calr_1$};		
							\draw (10, -2.7) -- (10, -3.3) node [above]{};	 				
							\draw[dim] (10,-3) -- (20,-3) node[midway,above] {$\calr_2$};
							\draw (20, -2.7) -- (20, -3.3) node [above]{};	 						
							\draw[dim] (20,-3) -- (25,-3) node[midway,above] {$\calr_3$};
							\draw (25, -2.7) -- (25, -3.3) node [above]{};	 	
							\draw[dim] (25,-3) -- (40,-3) node[midway,above] {$\calr_4$};       	
							
							\draw[fill=black] (13, 0) circle (0.4) node [above, yshift=0.1cm] {$x_1$};
							\draw[fill=black] (22.5, 0) circle (0.4) node [above, yshift=0.1cm] {$x_2$};		
							\end{tikzpicture}
							
							\caption{The 4 ranking zones. Every border $b_{i,j}$ is the midpoint between candidates $\calc_i,\calc_j$. The points $x_1,x_2$ have different rankings  -- the ranking of $x_1$ is $\pi_2=\calc_2 \succeq \calc_1 \succeq \calc_3$ whereas the ranking of $x_2$ is $\pi_3=\calc_2 \succeq \calc_3 \succeq \calc_1$. However, both strictly prefer $\calc_2$ over any other candidate, therefore for any consistent function $f$: $f(\pi_2)=f(\pi_3)=\calc_2$.}
							\label{fig:4-subzones}		
						\end{minipage}}} \par\setlength{\fboxrule}{0.2pt}

			\end{figure}
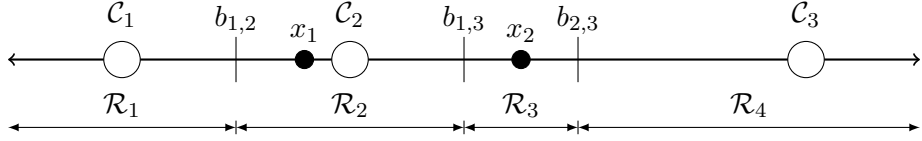
			
			$M$ is truthful --- if $\calc_1$ is chosen, then both agents prefer it over $\calc_3$ and have no incentive to misreport. If $\calc_3$ is chosen, then at least one of the agents is in zones $\calr_3$ or $\calr_4$ - this agent has no incentive to misreport since she prefers $\calc_3$ over $\calc_1$. The other agent has no influence over the outcome, therefore also has no incentive to misreport.			
					
			$M$ is not reducible to any voting mechanism --- any consistent function $f$ must map both $\pi_2$ and $\pi_3$ to $\calc_2$. However, whilst $M$ acts differently under inputs $(\pi_1,\pi_2)$ and $(\pi_1,\pi_3)$, then if there were a reduction, then both of these would have been mapped to the same location $M'(\calc_1,\calc_2)$ in contradiction.		
		\end{proof}
		
		
		\begin{claim} \label{claim:det-location-ranking}
			The set of truthful deterministic location mechanisms is equivalent to the set of truthful deterministic ranking mechanisms.
		\end{claim}
		\begin{proof}
			According to Lemma \ref{lemma:reducible-granular}, every truthful deterministic ranking mechanism is reducible to some truthful deterministic location mechanism. 
			It is left to show that every truthful deterministic location mechanism $M$ in $\Re^d$ is reducible to a truthful deterministic ranking mechanism $M'$.	
			
			The proof consists of several parts. 
			For an arbitrary location profile $\bx$ and an arbitrary truthful deterministic mechanism $M$, we define a location profile $\bx'$ which has no locations on borders, and show that it necessarily holds that $M(\bx)=M(\bx')$. 
			We then define a different location profile $\bx''$ and show the same: $M(\bx)=M(\bx'')$.
			The profile $\bx''$ is special in the sense that it is uniquely defined by some ranking profile $\mathbf{\pi}$.
			Finally, we show that given some input $\mathbf{\pi}=f_M(\bx)$ a consistent function $f_M$ (a function which depends on $M$, and will be defined later), there exists a ranking mechanism $M'$ which simulates $M$ on $\bx''$ (see Figure \ref{fig:nested-reductions}).
			The result is a constructive reduction which maintains the same output as the original location mechanism $M(\bx)$, as needed.  

			\begin{figure}[t]			
				\setlength{\fboxrule}{0.5 pt}
				\noindent \fbox{\noindent\makebox[1\textwidth][c]
					{\begin{minipage}{1\textwidth}						
							\centering						
							\includegraphics[scale=0.7]{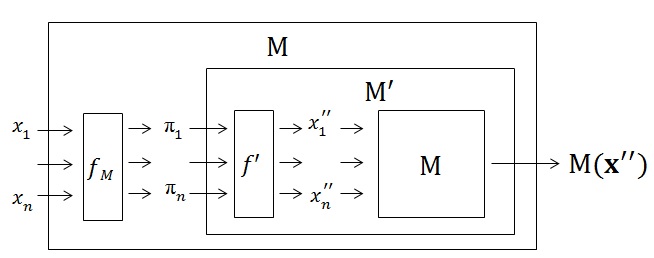}						
							\caption{A graphic sketch of the reductions in the proof of Claim \ref{claim:det-location-ranking}.}					
							\label{fig:nested-reductions}						
						\end{minipage}}} \par\setlength{\fboxrule}{0.2pt}
			\end{figure}			
			
			Let $M$ be an arbitrary truthful deterministic location mechanism, and let $\bx$ be an arbitrary location profile. 
			Let $M(\bx)=\calc_j$ for some candidate $\calc_j$ located at $y_j$. 
			Denote the ranking borders as $B$.
			
			Informally, we define $\bx'$ as a location profile which moves all agents in $\bx$ which are on a border, an infinitesimal distance towards the chosen candidate $\calc_j$.
			The resulting profile $\bx'$ has no agents on borders.
			We now define this formally: let $\epsilon$ be some small positive number. 
			For all $i$ such that $x_i \in B$, let $\vec{\epsilon_i}$ be a vector of size $\epsilon$  in direction $y_j-x_i$ \footnote[4]{If $x_i=y_j$ then the direction of $\vec{\epsilon_i}$ can be set arbitrarily. 
			}. 
			For any $i$ such that $x_i \notin B$, let $|\vec{\epsilon}_i|=0$.
			Let $\bx'=(x_1+\epsilon_1, \ldots x_n+\epsilon_n)$.
			We choose an $\epsilon$ sufficiently small such that for each $i$, $x_i'$ remains in the ranking zone of $x_i$.
			
			We now show that $M(\bx)=M(\bx')$ by moving agents from $\bx$ to $\bx'$ one by one, and showing that if any of these transitions were to change the chosen candidate, this would lead to a violation of truthfulness of $M$. Let:
			\begin{eqnarray*}		
				&& w_0 = (x_1,\ldots,x_n) = \bx \\
				&& w_1 = (x_1',x_2,\ldots,x_n) \\
				&& \ldots \\
				&& w_i = (x_1', \ldots, x_i',x_{i+1},\ldots,x_n) \\
				&& \ldots \\				
				&& w_n = (x_1',\ldots,x_n') = \bx'
			\end{eqnarray*}		
			Assume towards a contradiction that $M(w_0) \neq M(w_n)$. Let $i$ be the minimal index such that $M(w_{i-1}) \neq M(w_{i})$. It is known that $M(w_{i-1})=M(\bx)=\calc_j$. Denote $M(w_{i})=\calc_l$, where $\calc_l$ is located at point $y_l$. There are two options:
			\begin{itemize}
				\item If $|y_l-x_i| < |y_j-x_i|$, then in profile $w_{i-1}$, $x_i$ has an incentive to misreport to $x_i'$.
				\item If $|y_l-x_i| \geq |y_j-x_i|$, then $|y_j-x_i'|=|y_j-x_i|-\epsilon < |x_i'-y_l|$. Therefore in profile $w_i$, $x_i'$ has an incentive to misreport to $x_i$.
			\end{itemize}		
			Hence, $M(\bx)=M(\bx')$ as needed.
			
			Intuitively, we create the location profile $\bx''$ by moving all agents in $\bx'$ to some specific point within their ranking zone.
			Since $\bx'$ contained no agents on borders, each agent in $\bx'$ is located in exactly one ranking zone, hence $\bx''$ is well defined.
			We now define this formally:
			For any ranking zone $\calr_i$ such that $\calr_i \setminus B \neq \emptyset$, let $\hat{x}_i$ be some point in $\calr_i \setminus B$ (for instance, the centroid of the ranking zone $\calr_i$)\footnote[5]{We can safely disregard ranking zones which do not have any points which are not on a border, as no point in $\bx'$ will be in such a zone, since $\bx'$ does not contain any points on borders.}.
			Denote the ranking zone which contains $x_i'$ as $\calr_j$. 
			For all $i \in N$, let $x_i''=\hat{x}_j$. 
			Let $\bx'' = (x_1'',\ldots,x_n'')$.
			
			We now show that $M(\bx')=M(\bx'')$ in a similar fashion as we showed that $M(\bx)=M(\bx')$ previously.
			Let:
			\begin{eqnarray*}		
				&& h_0 = (x_1',\ldots,x_n') = \bx' \\
				&& h_1 = (x_1'',x_2',\ldots,x_n') \\
				&& \ldots \\
				&& h_i = (x_1'', \ldots, x_i'',x_{i+1}',\ldots,x_n') \\
				&& \ldots \\				
				&& h_n = (x_1'',\ldots,x_n'') = \bx''
			\end{eqnarray*}				
			Assume towards a contradiction that $M(h_0) \neq M(h_n)$. Let $i$ be the minimal index such that $M(h_{i-1}) \neq M(h_{i})$. Denote $M(h_{i})=\calc_m$ such that $\calc_m$ is located at $y_m$. Since $x_i',x_i''$ are in the same ranking zone and not on a border, there are two options:
			\begin{itemize}
				\item If $|y_m-x_i'| < |y_j-x_i'|$, then in profile $h_{i-1}$, $x_i'$ has an incentive to misreport to $x_i''$.
				\item If $|y_m-x_i'| > |y_j-x_i'|$, then it also holds that $|y_m-x_i''| > |y_j-x_i''|$, and in profile $h_i$, $x_i''$ has an incentive to misreport to $x_i'$.
			\end{itemize}		
			Therefore, $M(\bx')=M(\bx'') \Rightarrow M(\bx)=M(\bx'')$.
			
			It is left to show that it is possible to perform the ``nested reductions" as shown in Figure \ref{fig:nested-reductions}. Let $f_M$ be the consistent function which breaks ties just like $M$ does. That is, $f_M$ simulates $M$ on input $\bx$, finds the candidate $\calc_j$ and breaks ties in favor of rankings closer to $\calc_j$. 
			The output of $f_M$ is a ranking profile denoted by $\pi$.
			Given $\pi$, there exists some $M'$ that simulates $M$ on $\bx''$\footnote[6]{To avoid a circular definition, one can think of $M'$ simulating a location mechanism which acts precisely like $M$ does.} --- let $f'$ be a consistent function which maps every ranking $\pi_i$ (consistent with ranking zone $\calr_i$) to the point $\hat{x}_i$. 
			Therefore, such a reduction exists, and the output is $M(\bx'')=M(\bx)$. 
			
		\end{proof}
		
		\begin{claim}
			The set of truthful in expectation ranking mechanisms strictly contains the set of truthful in expectation voting mechanisms.				
		\end{claim}
		\begin{proof}
			As shown in Observation \ref{obs:reducible-granular}, any truthful in expectation voting mechanism is reducible to some truthful in expectation ranking mechanism.
			The proof of Claim \ref{lemma:det-ranking-voting} exhibits a truthful in expectation ranking mechanism which is not reducible to any voting mechanism.
		\end{proof}
		
		\begin{claim} \label{lemma:ran-location-ranking}
			The set of truthful in expectation randomized location mechanisms strictly contains the set of truthful in expectation randomized ranking mechanisms.
		\end{claim}
		\begin{proof}			
			As shown in Observation \ref{obs:reducible-granular}, any truthful in expectation ranking mechanism is reducible to some truthful in expectation location mechanism.
				
			We will show a truthful in expectation location mechanism $M$ which is not reducible to any ranking mechanism: Let there be 3 candidates at points $y_1=0,y_2=3,y_3=4$. $M$ acts as follows: 
			
			Choose an agent $i$ uniformly at random. Choose the candidates with the following probabilities:
			\begin{eqnarray*}				
				M(\bara) = 
				\begin{cases}
					\calc_1=1/3,\calc_2=1/3,\calc_3=1/3 \text{       if }  a_i \leq 1 \\
					\calc_1=1/4,\calc_2=1/2,\calc_3=1/4 \text{       otherwise.} 
				\end{cases}
			\end{eqnarray*}
			
			$M$ is not reducible to any ranking mechanism --- any consistent function $f$ must map both points $x_1 = 0.75$ and $x_2 = 1.25$ to $\pi_1= \calc_1 \succ \calc_2 \succ \calc_3$. 
			However, mechanism $M$ treats these two inputs differently. 
			
			It is left to show that $M$ is truthful in expectation. 
			We do so by assessing all possibilities of misreports.
			Obviously, the mechanism is not affected by any agents except the one who was chosen. 
			Since there are only two possible outcomes, it is sufficient to compare truthful reports $a_i$ with misreports $a_i'$ such that $a_i'$ changes the outcome. Let $\bara=(a_i,a_{-i})$ and $\bara'=(a_i',a_{-i})$.
			\begin{itemize}
				\item If $x_i \leq 0$ it holds that
				\begin{eqnarray*}
					\cost_{x_i}\left(M,\bara\right) &=&\frac{1}{3} \left[-x_i+(3-x_i)+(4-x_i)  \right]=-x_i+7/3 \\
					\cost_{x_i}\left(M,\bara'\right)&=& \frac{1}{4}\left[-x_i+(4-x_i)  \right]+ \frac{1}{2} \left( 3-x_i \right)=-x_i+5/2.
				\end{eqnarray*}
				Therefore: $\cost_{x_i}\left(M,\bara\right) \leq \cost_{x_i}\left(M,\bara'\right)$.
				
				\item If $0 < x_i < 3$ it holds that 
				\begin{itemize}
					\item If the outcome is $\calc_1=1/3,\calc_2=1/3,\calc_3=1/3$, the cost of agent $i$ is:
					\begin{eqnarray*}
						\frac{1}{3} \left[x_i+(3-x_i)+(4-x_i)  \right]= -x_i/3+7/3.
					\end{eqnarray*}					
					\item If the result is $\calc_1=1/4,\calc_2=1/2,\calc_3=1/4$ the cost of agent $i$ is:
					\begin{eqnarray*}
						\frac{1}{4}\left(x_i+4-x_i  \right)+ \frac{1}{2} \left(3-x_i \right)= -x_i/2+5/2.
					\end{eqnarray*}				
					It holds that $-x_i/3+7/3 \geq -x_i/2+5/2 \Leftrightarrow x_i \geq 1$. Therefore the first outcome is preferable to agents for which $0< x_i \leq 1$ and the second is better for agents for which $1<x_i<3$, and the mechanism is truthful in expectation in this interval.
				\end{itemize}
				
				\item If $3 \leq x_i <4$ it holds that 
				\begin{eqnarray*}
					\cost_{x_i}\left(M,\bara\right) &=& \frac{1}{4}\left(x_i+4-x_i  \right)+ \frac{1}{2} \left(x_i-3 \right)=x_i/2-1/2  \\
					\cost_{x_i}\left(M,\bara'\right) &=& \frac{1}{3} \left[x_i+(x_i-3)+(4-x_i)  \right]= x_i/3+1/3.
				\end{eqnarray*}				
				Therefore: $\cost_{x_i}\left(M,\bara\right) \leq \cost_{x_i}\left(M,\bara'\right) \Leftrightarrow x_i \leq 5$, therefore agent $i$ cannot benefit from misreporting.
				
				\item If $x_i \geq 4$ it holds that	
				\begin{eqnarray*}
					\cost_{x_i}\left(M,\bara\right) &=& \frac{1}{4}\left(x_i+x_i-4  \right)+ \frac{1}{2} \left(x_i-3 \right)=x_i-5/2  \\
					\cost_{x_i}\left(M,\bara'\right) &=& \frac{1}{3} \left[x_i+(x_i-3)+(x_i-4)  \right]=x_i-7/3.
				\end{eqnarray*}			
				Therefore: $\cost_{x_i}\left(M,\bara\right) \leq \cost_{x_i}\left(M,\bara'\right)$.
			\end{itemize}		
		\end{proof}		

		\begin{claim} \label{lemma:ran-voting-tie-ut}
			The set of truthful in expectation randomized voting mechanisms strictly contains the set of universally truthful randomized voting mechanisms.
		\end{claim}
		\begin{proof}	
		Any universally truthful voting mechanism is reducible to a truthful in expectation voting mechanism using the identity function $f$ (which is consistent).

		We exhibit a truthful in expectation (TIE) voting mechanism $M$ which is not reducible to any universally truthful mechanism. Let there be 2 candidates. $M$ chooses an agent $i$ uniformly at random, and chooses $a_i$ with probability 0.9 and the other candidate $\calc \setminus a_i$ with probability 0.1. 
		
		$M$ is truthful in expectation since for any agent $j$, if they are chosen, they are better off receiving their favorite candidate with probability 0.9 than with probability 0.1. $M$ is not universally truthful, since for each agent $i$ there exist cases in which reporting truthfully would lead to choosing their less favorite candidate, while there exist cases in which reporting non-truthfully would lead to choosing the favorite candidate. 
		Clearly, no composition with a consistent function $f$ can transform $M$ to a universally truthful mechanism (for instance, let $a_i = \calc_j$ for some $j$. From consistency, $f(a_i)=\calc_j$, so $f$ does not change the outcome at all).
		\end{proof}

		\begin{claim} \label{claim:two-candidates}
			When there are two candidates, the set of truthful in expectation randomized location mechanisms is equivalent to the set of truthful in expectation randomized voting mechanisms.		
		\end{claim}	
		\begin{proof}			
			As shown in Observation \ref{obs:reducible-granular}, any truthful in expectation voting mechanism is reducible to some truthful in expectation location mechanism.	
			This also holds for voting mechanisms for two candidates.
			We now show that any truthful in expectation location mechanism with two candidates is reducible to some truthful in expectation voting mechanism. 
			The proof follows similar lines as the proof of Claim \ref{claim:det-location-ranking}.
			
			Let $\bx$ be an arbitrary location profile, let $M$ be an arbitrary truthful in expectation location mechanism, and let $B$ be the border between $\calc_1$ and $\calc_2$. Define $\bx'$ as the location profile which moves all agents which are not on borders to their favorite candidate, that is:
			\begin{eqnarray*}				
				x_i' =  \begin{cases}
					y_1 & \quad \text{if }  x_i \in \calv_1 \setminus B \\
					y_2 & \quad \text{if }  x_i \in \calv_2 \setminus B \\
					x_i & \quad \text{if }  x_i \in  B 
				\end{cases}
			\end{eqnarray*}			
			We now show that $M(\bx)=M(\bx')$ by using a hybrid argument. Define:
			\begin{eqnarray*}		
				&& w_0 = (x_1,\ldots,x_n) = \bx \\
				&& w_1 = (x_1',x_2,\ldots,x_n) \\
				&& \ldots \\
				&& w_i = (x_1', \ldots, x_i',x_{i+1},\ldots,x_n) \\
				&& \ldots \\				
				&& w_n = (x_1',\ldots,x_n') = \bx'
			\end{eqnarray*}	
			
			Assume towards a contradiction that $M(w_0) \neq M(w_n)$. Then there exists some index $j$ such that $\Pr\left[M(w_j)=\calc_1\right] \neq \Pr\left[M(w_{j-1})=\calc_1\right]$. If this is the case then necessarily $x_j \notin B$ since that would imply that $w_{j-1}$ and $w_{j}$ are precisely the same profile. Assume without loss of generality that $\Pr\left[M(w_j)=\calc_1\right] > \Pr\left[M(w_{j-1})=\calc_1\right]$. There are 2 options:
			\begin{itemize}
				\item If $x_j,x_j' \in \calv_1$, then under location profile $w_j$, agent $j$ has an incentive to misreport to $x_j'$.
				\item If $x_j,x_j' \in \calv_2$, then under location profile $w_{j-1}$, agent $j$ has an incentive to misreport to $x_j$.
			\end{itemize}
			Therefore is necessarily holds that $M(\bx)=M(\bx')$.
			
			We now use $\bx'$ to show the reduction: 
			Let $f$ be a function which maps voting profiles to location profiles, by mapping each vote to candidate $\calc_i$ to location $y_i$.
			This function is clearly consistent.
			Let $M'$ be a voting mechanism which receives a voting profile, translates it to a location mechanism using the consistent function $f$, and then simulates $M$ on the output of $f$ (see Figure \ref{fig:two-candidates}).
			
			\begin{figure}[t]			
				\setlength{\fboxrule}{0.5 pt}
				\noindent \fbox{\noindent\makebox[1\textwidth][c]
					{\begin{minipage}{1\textwidth}						
							\centering						
							\includegraphics[scale=0.7]{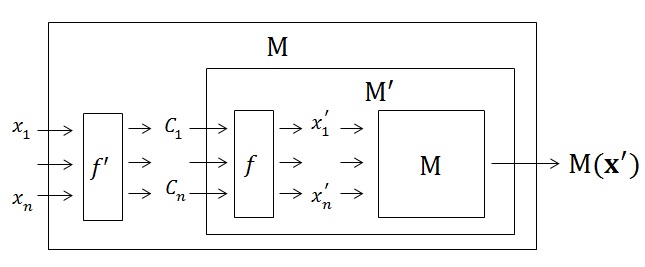}						
							\caption{A graphic sketch of the reductions in the proof of Claim \ref{claim:two-candidates}.}					
							\label{fig:two-candidates}						
						\end{minipage}}} \par\setlength{\fboxrule}{0.2pt}
			\end{figure}				
			
			For cases in which no agent is on the border, then the function $f'$ mapping location profiles to voting profiles is uniquely defined, and it holds that $M(\bx)=M'(f'(\bx))=M(f(f'(\bx)))$.
			It is left to show that in cases of agents on borders, there exists some consistent function $f'$ which breaks ties in the same manner that $M$ does.
			
			Let $\bx'^{(1)}$ be a location profile with $n_1$ agents at $y_1$, $n_2$ agents somewhere on the border $B$ and $n_3$ agents at $y_2$. In short, we note $\bx'^{(1)} = (n_1,n_2,n_3)$. 
			Recognize that $\bx'^{(1)}$ is a general location profile, after moving agents to their favorite candidates. 
			Using these amounts, define the following two location profiles:
			\begin{itemize}
				\item Let $\bx'^{(2)}$ be the profile in which there are $n_1+n_2$ agents at $y_1$ and $n_3$ agents at $y_2$ (that is $\bx'^{(2)}=(n_1+n_2,0,n_3)$).
				\item Let $\bx'^{(3)}$ be the profile in which there are $n_1$ agents at $y_1$ and $n_2+n_3$ agents at $y_2$ (that is $\bx'^{(3)}=(n_1,0,n_2+n_3)$).
			\end{itemize}
			Let $p_1=\Pr[M(\bx'^{(1)})=\calc_1]$, and similarly: $p_2=\Pr[M(\bx'^{(2)})=\calc_1]$ and $p_3=\Pr[M(\bx'^{(3)})=\calc_1]$. 
			Under these definitions, we show that: $p_3 \leq p_1 \leq p_2$:
			\begin{itemize}
				\item $p_1 \leq p_2$: Start with the profile $\bx'^{(1)}$, and move agents on the border one by one to $y_1$. If in each step the probability of choosing $\calc_1$ does not decrease then $p_1 \leq p_2$ as needed. Otherwise, there exists a profile $\hat{\bx}=(n_i,n_j,n_3)$ for which the probability of choosing $\calc_1$ is smaller than in the profile with $\hat{\bx'}=(n_i-1,n_j+1,n_3)$. If this were the case, then the agents on $y_1$ in profile $\hat{\bx}$ would benefit from misreporting to the point on a border, in contradiction to truthfulness.
				\item $p_3 \leq p_1$ is proved in the exact symmetrical manner, by moving agents from $B$ to $y_2$ one by one.
			\end{itemize}			
			
			By definition, a consistent function can map agents on borders to either of the two candidates, and can also choose any probabilities over the two agents. Therefore, for any $0 \leq q \leq 1$, there exists a consistent function which takes $n$ agents on the border and maps all of them to $\calc_1$ with probability $q$ and maps all of them to $\calc_2$ with probability $1-q$. For any $p_1$, we choose the consistent function $f'$ which uses a $q$ such that $p_2\cdot q + p_3\cdot(1-q)=p_1$. Under this function $f'$, the reduction does not change the outcome of the mechanism, as needed. 
			
		\end{proof}

	\end{proof}

	\subsection{Missing Proofs from Section \ref{sec-spike}}
	
	\begin{proof} of Lemma \ref{lemma:WPV-truthful}: \label{prf:wpv}
		Let $\bx$ be some location profile. Let $\bara$ be a voting profile, where $\bara$ is in ascending order, that is: $y(a_1) \leq y(a_2) \leq \ldots \leq y_(a_n)$. 
		Let $a_i'$ be a misreport of some agent $i$, and without loss of generality assume $y(a_i') < y(a_i)$. Denote the report directly to the left of $a_i'$ as $a_j$, that is: $y(a_j) \leq y(a_i') \leq y(a_{j+1})$ . Let $\bara'$ be the reports $(a_i',a_{-i})$ in ascending order (see Figure \ref{fig:WPV}).
		\begin{eqnarray*}		
			cost_{x_i}(M,\bara)-cost_{x_i}(M,\bara') &=& \sum_{k=1}^n p_k|y(a_k)-x_i| - \sum_{k=1}^n p_k|y(a_k')-x_i| \\
			&=& \sum_{k=j+1}^i p_k[|y(a_k)-x_i|-|y(a_k')-x_i|] 			
		\end{eqnarray*}		
		The last equality holds since for any $k$ such that $k \leq j$ or $k>i$: $y(a_k)=y(a_k')$. 
		
		\noindent For any $k$ such that: $j<k\leq i$: $|y(a_k')\leq y(a_k)|$. In this case:
		\begin{itemize}
			\item If $x_i \geq y(a_k)$: $|y(a_k')-x_i| \geq |y(a_k)-x_i|$ from the triangle inequality.
			\item If $x_i < y(a_k)$: Then $y(a_k)=y(a_i)$ and $a_i$ is the true action of $x_i$, i.e., there is no candidate closer to $x_i$ than it.			
		\end{itemize}
		Therefore, for all $k$: $|y(a_k)-x_i|\leq |y(a_k')-x_i|$. Hence the agent cannot benefit from misreporting regardless of the results of the random bits, so $M$ is universally truthful.
		\begin{figure} [t]
			
			\setlength{\fboxrule}{0.5 pt}
			\noindent \fbox{\noindent\makebox[1\textwidth][c]
				{\begin{minipage}{1\textwidth}
						
						\centering
						
						\begin{tikzpicture}[y=.3cm, x=.3cm,font=\sffamily]
						\draw[<-, thick] (0,0) -- (20,0);
						\draw[dashed, thick] (20,0) -- (26,0);
						\draw[->, thick] (26,0) -- (40,0);										
						
						\draw[fill=white] (5, 0) circle (0.4) node  (1) {};
						\draw[fill=white] (8, 0) circle (0.4) node  (2) {};
						
						\draw[fill=white] (12, 0) circle (0.4) node (3) {} ;       			
						\draw[->, thick] (32,1) .. controls (22,2) .. (13, 1);			
						
						\draw[fill=white] (16, 0) circle (0.4) node (4) {};
						\draw[fill=white] (28, 0) circle (0.4) node (5) {};
						\draw[fill=white] (33, 0) circle (0.4) node (6) {};
						\draw[fill=white] (38, 0) circle (0.4) node (7) {};
						
						\node [below = 0.3 cm of 1](a){$a_{j-1}$};
						\node [below = 0.8 cm of 1](b){$a_{j-1}$};			
						\node [below = 1.4 cm of 1](c){$a_{j-1}'$};			
						\node [below = 0.3 cm of 2]{$a_{j}$};
						\node [below = 0.8 cm of 2]{$a_{j}$};				
						\node [below = 1.4 cm of 2]{$a_{j}'$};							
						\node [below = 0.8 cm of 3]{$a_{i}'$};	
						\node [below = 1.4 cm of 3]{$a_{j+1}'$};				
						\node [below = 0.3 cm of 4]{$a_{j+1}$};
						\node [below = 0.8 cm of 4]{$a_{j+1}$};	
						\node [below = 1.4 cm of 4]{$a_{j+2}'$};
						\node [below = 0.3 cm of 5]{$a_{i-1}$};
						\node [below = 0.8 cm of 5]{$a_{i-1}$};
						\node [below = 1.4 cm of 5]{$a_{i}'$};			
						\node [below = 0.3 cm of 6]{$a_{i}$};
						\node [below = 0.3 cm of 7]{$a_{i+1}$};
						\node [below = 0.8 cm of 7]{$y(a_{i+1})$};
						\node [below = 1.4 cm of 7]{$y(a_{i+1}')$};
						
						\node [left = 0.5 cm of a, align=center]{$\bara$:};
						\node [left = 0.5 cm of c, align=center ]{$\bara'$:};
						\end{tikzpicture}
						
						\caption{Misreporting in a WPV: $\bara$ is shown in the first line, the misreport in the second line, and the misreport after renaming the reports in ascending order in the third line.}						
						\label{fig:WPV}
					\end{minipage}}} \par\setlength{\fboxrule}{0.2pt}

				\end{figure}
				
			\end{proof}

	\begin{proof} of Lemma \ref{lemma:tight}: \label{prf:tight}
		For any candidate $j$ denote $p_j=\Pr[M(\bara)=y_j]$ and $p_j'=\Pr[M(\bara')=y_j]$.
		Define $\Delta_j$ as the difference in the cost of candidate $j$ under profile $\mathbf{x}$ and her cost under  $\mathbf{x}'$, that is:
		$\Delta_j = \sum_{i=1}^n |y_j-x_i|- \sum_{i=1}^n |y_j-x_i'|$. 
		Since $\bx'$ was defined by moving all agents towards $\calc_{opt}$, then: $\forall j: \Delta_{opt} \geq \Delta_j$. As noted previously, this means that $\calc_{opt}$ remains the optimal candidate under $\bara'$. According to Lemma \ref{lemma:border-outward}, the worst-case ratio occurs when all agents on borders vote outwards (farther from $y_{opt}$), so the probabilities under $\bara'$ remain the same as under $\bara$: $\forall j: p_j=p_j'$.
		
		We now assess the approximation ratio under profile $\bx'$:
		$$\SC(\OPT,\bx') = \sum_i |y_{opt}-x_i'|=\SC(\OPT,\bx)-\Delta_{opt}$$
		The cost of the spike mechanism on $\bx'$ is:
		\begin{eqnarray*}				
			\SC(M,\bx') &=& p_j' \cdot \left[\sum_i |y_j-x_i'|\right] \\
			&=& p_j \cdot \left[ \left( \sum_i |y_j-x_i| \right) -\Delta_j\right] \\
			&\geq& p_j \cdot \left[ \left(  \sum_i |y_j-x_i| \right) -\Delta_{opt}\right] \\
			&=& \SC(M,\bx)-\Delta_{opt}
		\end{eqnarray*}	
		Therefore, the new approximation ratio is:
		\begin{eqnarray*}				
			\displaystyle \frac{\SC(M,\bx')}{\SC(\OPT,\bx')} &\geq&
			\displaystyle \frac{\SC(M,\bx') - \Delta_{opt}} {\SC(\OPT,\bx')-\Delta_{opt}} 
			\geq \displaystyle \frac{\SC(M,\bx)}{\SC(\OPT,\bx)}
		\end{eqnarray*}			
	\end{proof}		
	
	\begin{proof} of Lemma \ref{lemma:inwards}: \label{prf:inwards}
		We use the same notation as in Figures \ref{fig:spike-2} and \ref{fig:spike-3} where there are $n_i$ agents at point $\hat{x}_i$.
		
		Let $\Delta=\SC(\OPT,\bx)-\SC(\OPT,\bx') = n_1\cdot(\hx_2-\hat{x}_1)>0$. 
		It is sufficient to show that $\SC(S,\bx, \bara)-\SC(S,\bx', \bara) \leq 2\Delta$ since that would imply:
		\begin{eqnarray*}				
			\frac{\SC(S,\bx, \bara)}{\SC(\OPT,\bx)} \leq 
			\frac{\SC(S,\bx', \bara)+2\Delta}{\SC(\OPT,\bx')+\Delta} \leq \frac{2\cdot \SC(\OPT,\bx'))+2 \Delta}{\SC(\OPT,\bx')+\Delta}=2
		\end{eqnarray*}
		Denote the probabilities as follows: $p_i=\Pr[S(\bara)=y_i]$, $p_i'=\Pr[S(\bara')=y_i]$. Similarly, the costs of the candidates are denoted: $c_i= \SC (\calc_i,\bx)= \sum_{j=1}^n|x_j-y_i|$ and $c_i'= \SC (\calc_i,\bx')=\sum_{j=1}^n|x_j'-y_i|$. 
		The probabilities and costs under profile $\bx'$ are:
		\begin{eqnarray*}				
			p_i' = \begin{cases}
				0 & \quad \text{if }  i=1 \\
				p_1+p_2 & \quad \text{if }  i=2 \\
				p_i & \quad \text{if }  i \geq 3
			\end{cases}
		\end{eqnarray*}			
		Denote $\delta = n_1\left(|\hat{x}_2-y_2|-|y_2-\hat{x}_1|\right)$, so the costs are:
		\begin{eqnarray*}				
			c_i' = \begin{cases}
				c_2+n_1\left(|x_2-y_2|-|y_2-x_1|\right)=c_2+\delta & \quad \text{if }  i=2 \\
				c_i-\Delta & \quad \text{if }  i \geq 3
			\end{cases}
		\end{eqnarray*}				
		
		Therefore, the difference in the cost is:
		\begin{eqnarray*}				
			\SC(S,\bx)-\SC(S,\bx') &=& \sum_i \left(p_i c_i - p_i' c_i'\right) \\
			&=& p_1 c_1 + p_2 c_2 - \left( p_1 + p_2 \right)\left(c_2 +\delta\right)  + \sum_{i \geq 3}p_i \left[c_i - (c_i-\Delta )\right] \\
			&=& p_1\left(c_1-c_2-\delta\right) - p_2\cdot \delta +   \sum_{i \geq 3}p_i\Delta \\
			&=& p_1\left(c_1-c_2-\delta\right) - p_2\cdot \delta + \left(1 - p_1 - p_2\right)\Delta
		\end{eqnarray*}			
		Due to the triangle inequality: 
		\begin{eqnarray*}				
			&& |\hx_2-y_2| \leq |\hx_2 - \hx_1| + |\hx_1 - y_2| \\
			&\Leftrightarrow& \delta = n_1 \left( |\hx_2 - y_2| - |\hx_1 - y_2|  \right) \leq n_1 |\hx_2 - \hx_1| = |\Delta|
		\end{eqnarray*}				
		
		Therefore:
		\begin{eqnarray*}				
			&& p_1\left(c_1-c_2-\delta\right) - p_2\cdot \delta + \left(1 - p_1 - p_2\right)\Delta \\
			&\leq& p_1\left(c_1-c_2+\Delta \right) + p_2\cdot \Delta + \left(1 - p_1 - p_2\right) \Delta \\
			&=& p_1\left(c_1-c_2 \right) + \Delta
		\end{eqnarray*}													
		Also $c_1-c_2=(n-n_1)|y_2-y_1|$, so together:
		\begin{eqnarray*}				
			\SC(S,\mathbf{x})-\SC(S,\mathbf{x}') &\leq& p_1\left[(n-n_1)|y_2-y_1| \right] + \Delta
		\end{eqnarray*}		
		To conclude the proof it is left to show that $p_1\left[(n-n_1)|y_2-y_1| \right] \leq \Delta=|\hat{x}_2-\hat{x}_1|n_1$. 
		Since $\frac{|y_2-y_1|}{2} = |y_2-\hat{x}_1|<|\hat{x}_2-\hat{x}_1|$, it is sufficient to show that: $p_1(n-n_1)\leq \frac{n_1}{2}$. 
		$S$ is a spike mechanism, so $p_1$ is:
		\begin{itemize}
			\item If $n_1 \leq n/2$ then: $p_1=\frac{n_1}{2(n-n_1)} \Rightarrow p_1(n-n_1)=\frac{n_1}{2(n-n_1)}(n-n_1)=n_1/2$.
			\item If $n_1 > n/2$ then:  $p_1=1.5-\frac{n}{2n_1}=\frac{3n_1-n}{2n_1}$, so:
			\begin{eqnarray*}				
				p_1(n-n_1) &=& \frac{3n_1-n}{2n_1}(n-n_1) = \frac{-n^2-3n_1^2+4nn_1}{2n_1} = \frac{-n^2-4n_1^2+4nn_1}{2n_1}+\frac{n_1^2}{2n_1} \\
				&=& \frac{-(n-2n_1)^2}{2n_1}+\frac{n_1}{2} \leq \frac{n_1}{2}.
			\end{eqnarray*}	
			This concludes the proof of the lemma.
		\end{itemize}		
	\end{proof}	
	
	\begin{proof} of Lemma \ref{lemma:three-candidates}: \label{prf:three-candidtes}
		We use the notations of the location of the left, center and right candidates in the following manner $y_L=y_{opt-1}$, $y_C=y_{opt}$, $y_R=y_{opt+1}$, and the amount of agents in $b_L,y_C,b_C$ as $L,C,R$ respectively (see Figure \ref{fig:spike-4}). Denote the probabilities of choosing the candidates as $p_L=\Pr(S(\bx)=y_L)$, $p_C=\Pr(S(\bx)=y_C)$, $p_R=\Pr(S(\bx)=y_R)$.  Also, without loss of generality, the distances can be scaled such that $b_C-y_C=1$. Define: $\beta = y_C-b_L$.
		
		The costs of the different candidates are:
		\begin{eqnarray*}		
			&& \SC(\calc_L, \bx) = \beta \cdot (L+2C+2R)+R  \\
			&& \SC(\calc_C, \bx) = L\cdot \beta + R =\SC(\OPT,\bx) \\
			&& \SC(\calc_R, \bx) = L\beta + (2L+2C+R)
		\end{eqnarray*}	
		
		Due to the definition of the spike mechanism, the proof is broken into two parts:
		\begin{enumerate}
			\item The median agent is on $y_C$
			\item The median agent is on $b_L$
		\end{enumerate}
		Note that the last option (in which the median agent is on $b_C$) is identical to the second case due to symmetry, therefore proving for these two cases is sufficient.
		
		In the first case, the median agent is at the center, therefore $L<C+R$ and $R<L+C$, and from the definition of the spike mechanism:
		\begin{eqnarray*}		
			&& p_L = \frac{L}{2(C+R)} \\
			&& p_R = \frac{R}{2(C+L)} \\
			&& p_C = 1- p_L - p_R = 1 - \frac{L}{2(C+R)} - \frac{R}{2(C+L)}
		\end{eqnarray*}			
		Therefore the ratio is:
		\begin{eqnarray*}		
			\frac{\SC(M,\bx)}{\SC(\OPT,\mathbf{x})} &=&  \frac{p_L\SC(\calc_L, \bx) + p_C\SC(\calc_C, \bx) + p_R\SC(\calc_R, \bx)} {\SC(\calc_C, \bx)} \\
			&=& \frac{p_L\SC(\calc_L, \bx) + p_R\SC(\calc_R, \bx)} {\SC(\calc_C, \bx)} + p_C \\
			&=& \frac{1}{L\beta+R} \left[ \frac{L(\beta(L+2C+2R)+R )}{2(C+R)}+\frac{R(L\beta + (2L+2C+R) )}{2(L+C)}  \right] \\
			&+& \left( 1 - \frac{L}{2(C+R)} - \frac{R}{2(L+C)} \right) \\
			&=& 1 + \frac{L(L\beta+2C\beta+2R\beta+R )}{2(C+R)(L\beta+R)} - \frac{L(L\beta+R)}{2(C+R)(L\beta+R)} \\
			&+& \frac{R(L\beta+2L+2C+R)}{2(L+C)(L\beta+R)} - \frac{R(L\beta+R)}{2(L+C)(L\beta+R)} \\
			&=& 1 + \frac{L(2C\beta+2R\beta)}{2(C+R)(L\beta+R)} + \frac{R(2L+2C)}{2(L+C)(L\beta+R)} \\
			&=& 1 + \frac{L\beta}{L\beta+R}+ \frac{R}{L\beta+R} = 2
		\end{eqnarray*}				
		Therefore the ratio cannot exceed 2 in the first case. 
		
		In the second case the median agent is at $b_L$ and the probabilities are:
		\begin{eqnarray*}		
			p_L &=& 1.5 - \frac{L+C+R}{2L} = 1 - \frac{C+R}{2L} \\
			p_R &=& \frac{R}{2(C+L)} \\
			p_C &=& 1- p_L - p_R = 1 - \left( 1 - \frac{C+R}{2L} \right) - \frac{R}{2(C+L)} \\
			&=& \frac{C+R}{2L} - \frac{R}{2(C+L)}
		\end{eqnarray*}				
		Therefore the approximation ratio is:
		\begin{eqnarray*}		
			&& \frac{\SC(M,\bx)}{\SC(\OPT,\bx)} =   \frac{p_L\SC(\calc_L,\bx)+p_R\SC(\calc_R,\bx)}{\SC(\calc_C,\bx)}+p_C \\
			&=& \frac{1}{L\beta+R} \left[\left(1 - \frac{C+R}{2L}  \right) (L\beta+2C\beta+2R\beta+R) + \frac{R}{2(L+C)}(L\beta+2L+2C+R) \right] \\
			&+& \frac{C+R}{2L} - \frac{R}{2(C+L)} \\
			&=& \frac{1}{L\beta+R} \left[ L\beta+R+2C\beta+2R\beta - \frac{(C+R)(L\beta+2C\beta+2R\beta+R)}{2L} + \frac{R(2L+2C+L\beta+R)}{2(L+C)}  \right] \\
			&+& \frac{C+R}{2L} - \frac{R}{2(C+L)} \\
			&=& 1 + \frac{1}{L\beta+R} \left[ 2\beta(C+R) - \frac{(C+R)L\beta}{2L}-\frac{(C+R)(2C\beta+2R\beta+R)}{2L} + R +  \frac{R(L\beta+R)}{2(L+C)}  \right] \\
			&+& \frac{C+R}{2L} - \frac{R}{2(C+L)} \\
			&=& 1 + \frac{1}{L\beta+R} \left[ \frac{3\beta(C+R)}{2} - \frac{(C+R)(2C\beta+2R\beta+R)}{2L}  + R   \right]  + \frac{R}{2(L+C)} \\
			&+& \frac{C+R}{2L} - \frac{R}{2(C+L)} \\
			&=& 1 + \frac{3\beta(C+R)+2R}{2(L\beta+R)} -  \frac{(C+R)(2C\beta+2R\beta+R)}{(L\beta+R)2L}  + \frac{LR}{2L(L+C)} + \frac{C+R}{2L} - \frac{R}{2(C+L)}   \\
			&=& 1 + \frac{3\beta(C+R)+2R}{2(L\beta+R)} - \frac{(C+R)(2C\beta+2R\beta+R)}{2L(L\beta+R)}  +  \frac{(C+R)}{2L}   \\
		\end{eqnarray*}		
		Now, in order to show this is a 2 approximation:
		\begin{eqnarray*}		
			&& 1 + \frac{3\beta(C+R)+2R}{2(L\beta+R)} - \frac{(C+R)(2C\beta+2R\beta+R)}{2L(L\beta+R)}  +  \frac{(C+R)}{2L}  \leq 2 \\
			&\Leftrightarrow& \frac{3\beta(C+R)+2R}{2(L\beta+R)} - \frac{(C+R)(2C\beta+2R\beta+R)}{2L(L\beta+R)}  +  \frac{(C+R)(L\beta+R)}{2L(L\beta+R)}  \leq 1
		\end{eqnarray*}			
		And by multiplying both sides by the common denominator $2L(L\beta+R)$:
		\begin{eqnarray*}		
			&& L[3\beta(C+R)+2R] + (C+R)(-2C\beta-2R\beta-R+L\beta+R)  \leq 2L(L\beta+R) \\
			&\Leftrightarrow&  L[3\beta(C+R)+2R] + (C+R)(L\beta-2C\beta-2R\beta)  \leq 2L(L\beta+R) \\
			&\Leftrightarrow& L[3\beta(C+R)] + \beta(C+R)(L-2C-2R)  \leq 2L \cdot L\beta
		\end{eqnarray*}				
		Since $\beta$ is always positive, it is possible to divide both sides by $\beta$:
		\begin{eqnarray*}		
			&& L[3(C+R)] + (C+R)(L-2C-2R)  \leq 2L^2 \\
			&\Leftrightarrow& 3LC+3LR + LC-2C^2-2CR+LR-2CR-2R^2  \leq 2L^2 \\
			&\Leftrightarrow& 0  \leq 2L^2+2C^2+2R^2-4LC-4LR+4CR \\
			&\Leftrightarrow& 0  \leq L^2+C^2+R^2-2LC-2LR+2CR \\
			&\Leftrightarrow& 0  \leq (L-C-R)^2 \\			
		\end{eqnarray*}						
		This term is indeed always non-negative, so this concludes the proof.
	\end{proof}			
	
	\begin{lemma} \label{lemma:border-outward}
		Let $\mathbf{x}$ be an arbitrary location profile, and let agent $i$ be on a border $b_j$ such that $y_{j}<x_i=b_j<y_{j+1} \leq y_{opt}$. 
		
		\noindent Let $\bara_1=(a_i=\calc_{j},\bara_{-i})$, let $\bara_2=(a_i=\calc_{j+1},\bara_{-i})$, and let $M$ be some WPV mechanism. 
		
		\noindent Then $\SC(M,\bx,\bara_1) \geq \SC(M,\bx,\bara_2)$.
	\end{lemma}	
	\begin{proof}
		Let $p_i = \Pr[M(\bara_1)=\calc_i]$ and $q_i = \Pr[M(\bara_2)=\calc_i]$. 
		According to the definition of WPV mechanisms, the change of vote only affects the probabilities of candidates $\calc_j$ and $\calc_{j+1}$, that is: $\forall k \neq j, j+1$: $p_k=q_k$. 
		Denote: $p_{j}=q_{j}+\alpha$ and $p_{j+1}+\alpha=q_{j+1}$ for some $\alpha>0$. 
		Therefore:
		\begin{eqnarray*}		
			\SC(M,\bx, \bara_1)-\SC(M,\bx, \bara_2) &=& 
			\sum_k p_k \cdot \SC(\calc_k,\bx) - \sum_k q_k \cdot  \SC(\calc_k,\bx) \\
			&=& \alpha \left[\SC(\calc_{j},\bx)-\SC(\calc_{j+1},\bx)\right]
		\end{eqnarray*}	
		
		Therefore it is sufficient to show that $\SC(\calc_{j},\bx) \geq \SC(\calc_{j+1},\bx)$. We define the cost function for any point on the line: $f(x) = \sum_k |x-x_k|$. By definition, for any candidate $\calc_l$: $\SC(\calc_{l},\bx) = f(y_l)$. 
		
		Clearly, $f(x)$ is single-peaked, with a peak at the median $x_{\ceil*{n/2}}$, since moving in any direction away from the median only increases the distance to at least half of the agents. 
		
		We check the different cases:
		\begin{itemize}
			\item If $y_{opt} \leq \xmed$: Then $y_j < y_{j+1} \leq y_{opt} \leq \xmed \Leftarrow f(y_j) \geq f(y_{j+1})$.
			\item If $y_{opt} > \xmed$: 
			\begin{itemize}
				\item If $y_{j+1}=y_{opt}$: Then the proof is concluded by definition of optimality.
				\item If $y_{j+1} < y_{opt}$ then by definition of optimality: $y_j < y_{j+1} <\xmed < y_{opt}$. Therefore: $f(y_j) \geq f(y_{j+1})$.
			\end{itemize}
		\end{itemize}		
	\end{proof}			
	The proof also holds for the symmetrical case in which $y_{opt} \leq y_j < b_j = x_i < y_{j+1}$.

	\subsection{Missing Proofs from Section \ref{sec-randomized}}
	
	\begin{proof} of Lemma \ref{lemma:border-equal}: \label{prf:border-equal}
		Proof via contradiction. 
		Assume $\cost_{x_l}(M,(a_l=\pi_j,a_{-l})) = \cost_{x_l}(M,(a_l=\pi_i,a_{-l}))+\delta$ for some $\delta>0$.
		Let there be an agent $k$ located in ranking zone $\calr_j$ such that $|x_k-x_l|=\epsilon < \frac{\delta}{2}$. 
		
		Then agent $k$ has an incentive to misreport:
		\begin{eqnarray*}		
			\cost_{x_k}(M,(a_k=\pi_j,a_{-k})) &\geq&  \cost_{x_l}(M,(a_k=\pi_j,a_{-k})) - \epsilon \\
			&=& \cost_{x_l}(M,(a_k=\pi_i,a_{-k})) + \delta - \epsilon \\
			&>& \cost_{x_l}(M,(a_k=\pi_i,a_{-k})) + \epsilon \\
			&\geq& \cost_{x_k}(M,(a_k=\pi_i,a_{-k}))
		\end{eqnarray*}		
		The transitions in the first and last rows are due to the triangle inequality (for any location the mechanism may choose), the second row holds by the assumption, and the third row holds since $\epsilon < \frac{\delta}{2}$.
		
		Agent $k$ has an incentive to misreport, contradicting the assumption and completing the proof.
	\end{proof}

	\begin{proof} of Lemma \ref{lemma:rand-lb-2-non-strategic}: \label{prf:rand-lb-2-non-strategic}
		We prove the lower bounds for the case of two agents and two candidates.
		Let $y_1=-1,y_2=1$.
		Let $\bara=(a_1,a_2)$ be the ranking profile in which the two agents prefer different candidates, that is: $a_1=\calc_1 \succeq \calc_2$, $a_2 = \calc_2 \succeq \calc_1$.
		
		Examine the following two location profiles $\bx,\bx'$, (in both cases for which $\bara$ is a truthful voting profile): $\mathbf{x}=(-1,\epsilon)$, $\mathbf{x'}=(-\epsilon,1)$.
		
		We show that any decision of the mechanism makes will cause an approximation ratio of 2 either in $\bx$ or in $\bx'$.
		
		It is easy to see that:
		\begin{eqnarray*}
			&& \SC(\calc_1,\bx) = 1 + \epsilon = \SC(\OPT,\bx) \\
			&& \SC(\calc_2,\bx) = 3 - \epsilon \\			
			&& \SC(\calc_1,\bx') = 3 - \epsilon \\			
			&& \SC(\calc_2,\bx') = 1 + \epsilon = \SC(\OPT,\bx')
		\end{eqnarray*}

		Denote $p  =\Pr[M(\bara)=\calc_1]$. Therefore:
		\begin{eqnarray*}
			&& \SC(M,\bx) = p(1 + \epsilon) + (1-p)(3-\epsilon) \\
			&& \SC(M,\bx') = p(3 - \epsilon) + (1-p)(1 + \epsilon) 
		\end{eqnarray*}			 
		
		The approximation ratio is therefore at least: 
		
		\begin{eqnarray*}		
			&& \min_{0 \leq p \leq 1} \left\lbrace \max \left\lbrace \frac {\SC(M,\bx) }{\SC(\OPT,\bx)}, \frac{\SC(M,\bx')}{\SC(\OPT,\bx')} \right\rbrace \right\rbrace = \\
			&& \min_{0 \leq p \leq 1} \left\lbrace \max \left\lbrace \frac {p(1+\epsilon)+(1-p)(3-\epsilon)}{1+\epsilon}, \frac{(1-p)(1+\epsilon)+p(3-\epsilon)}{1+\epsilon} \right\rbrace \right\rbrace = \\
			&& \min_{0 \leq p \leq 1} \left\lbrace \max \left\lbrace 1+2p-p\epsilon, 3-2p+2p\epsilon-\epsilon \right\rbrace \right\rbrace
		\end{eqnarray*}
		The optimal value is reached at $p=0.5$, and it is $2-\frac{\epsilon}{2}$, which tends to 2 as $\epsilon$ tends to $0$.
	\end{proof}

	\begin{proof} of Lemma \ref{lemma:rd-3} (Random Dictator): \label{prf:rd-3}
		Random dictator (RD) is a WPV mechanism and so it is universally truthful according to Lemma \ref{lemma:WPV-truthful}.
		
		We start by showing that the ratio can be arbitrarily close to 3. Let $y_1=-1,y_2=1$, and let $x_1=\ldots=x_{n-1}=-1$ and $x_n=1$. Therefore the costs are: $\SC(\calc_1, \bx)=1+\epsilon=\SC(\OPT,\bx)$ and $\SC(\calc_2, \bx) = 2(n-1)+(1-\epsilon)$. 
		
		Ergo: $\SC(RD,\bx)= \frac{n-1}{n}\cdot \SC(\calc_1,\bx) + \frac{1}{n}\cdot \SC(\calc_2,\bx)=3-\frac{2}{n}+\frac{2\epsilon}{n}+\epsilon$. 
		The approximation ratio is therefore:
		\begin{eqnarray*}		
			\frac{\SC(RD,\bx)}{\SC(\OPT,\bx)} = \frac{3-\frac{2}{n}+\frac{2\epsilon}{n}+\epsilon}{1+\epsilon} 
			= 3-\frac{2\epsilon+\frac{2}{n}-\frac{2\epsilon}{n}}{1+\epsilon}
		\end{eqnarray*}		
		
		Clearly this ratio tends to $3$ as $n \rightarrow \infty, \epsilon \rightarrow 0$.
		
		We now show that the approximation ratio is bounded from above by 3.
		The social cost is:
		\begin{eqnarray*}		
			&& \SC(RD,\bx)=\frac{1}{n} \sum_{i} \left( \sum_j |x_j-y(a_i)| \right) \\
			&\leq& \frac{1}{n} \sum_{i} \left( \sum_j |x_j-y_{opt}| +|y_{opt}-y(a_i)| \right) \\
			&\leq& \frac{1}{n} \sum_{i} \left( \sum_j |x_j-y_{opt}| +|y_{opt}-x_i|+|x_i-y(a_i)| \right) \\
			&\leq& \frac{1}{n} \sum_{i} \left( \sum_j |x_j-y_{opt}| +2|y_{opt}-x_i| \right) \\		
			&=& \frac{1}{n} \sum_{i} \left( \sum_j |x_j-y_{opt}| \right) + \frac{1}{n} \sum_{i} \left( \sum_j 2|y_{opt}-x_i| \right) \\
			&=& \frac{1}{n} \sum_j \left(|x_j-y_{opt}|\sum_{i}(1) \right) + \frac{2}{n} \sum_{i} \left(|y_{opt}-x_i| \sum_j (1) \right) \\	
			&=& \frac{1}{n} \sum_j \left(|x_j-y_{opt}|\cdot n \right) + \frac{2}{n} \sum_{i} \left(|y_{opt}-x_i| \cdot n \right)	\\			
			&=& \SC(\OPT,\bx) + 2\cdot \SC(\OPT,\bx) = 3\cdot \SC(\OPT,\bx)
		\end{eqnarray*}		
		
		The first two transitions hold due to the triangle inequality, and the third inequality holds due to fact that no candidate is closer to $x_i$ than $y(a_i)$ is.
		
		Notice that this holds in any metric space since we only used the triangle inequality, and did not use any notion which is specific to the line.
	\end{proof}

	\label{GSP}
	A mechanism is \textit{group-strategyproof (GSP)} if for any location profile and any coalition $S \subseteq N$, there is no joint deviation of the agents in $S$ from the truthful reports such that they all gain.
	That is: 
	\begin{eqnarray*}		
		&& \forall S\subseteq N, \forall a_S \in \cala(\bx_S), \forall a_{-S} \in \cala^{n-|S|}, \forall a_S' \in \cala^{|S|}, \exists i \in S: \\
		&& \cost_{x_i}(M,(a_S,a_{-S})) \leq \cost_{x_i}(M,(a_S',a_{-S}))
	\end{eqnarray*}

	In the continuous model, random dictator is GSP on the line (and even on the circle, see \cite {alon2010walking}). We show that in our candidate model random dictator is not GSP on the line.
	Notice that random dictator is in particular a WPV mechanism, therefore a corollary of the lemma is that there exist WPV mechanisms which are not GSP.
	\begin{lemma} \label{lemma:GSP}
		Random dictator is not group group-strategyproof
	\end{lemma}
	\begin{proof}
		Let there be three candidates at locations $y_1 = 1,y_2 = 0, y_3 = 1$ and let there be two agents at $x_1=-0.51$ and $x_2=0.51$.
		When both agents report truthfully ($a_1 = \calc_1$, $a_2 = \calc_3$), the mechanism chooses $\calc_1,\calc_3$, each with probability 0.5.
		The cost of each of the agents in this case is: $\cost_{x_1}(RD,\bara)=\cost_{x_2}(RD,\bara)=0.5\cdot (0.51+1.49)=1$.
		
		However, if both agents misreport together to $\bara'=(a_1'=\calc_2,a_2'=\calc_2)$, then $\calc_2$ will always be chosen. The costs in this case will be: $\cost_{x_1}(RD,\bara')=\cost_{x_2}(RD,\bara')=0.51$.
	\end{proof}		
	
	\subsection{Missing Proofs from Section \ref{sec-deterministic}}
	
\begin{proof} of Lemma \ref{lemma:det-med-3}: \label{prf:det-med-3}
	This mechanism is truthful - any agent located at the median location has no incentive to misreport since the only possible consequence is for the mechanism to select a different location. Similarly, other agents have no incentive to misreport, since misreporting either has no effect or moves the chosen location further away. 
	
	We now move on to the approximation ratio - 
	Let $\pi$ be the permutation of $\{1 \ldots n\}$ such that $y \left( a_{\pi(i)} \right) \leq \left(a_{\pi(i+1)}\right)$ for each $i \in \{1 \ldots n-1\}$.
	Denote $\hat{a}_i=a_{\pi(i)}$.
	The median mechanism chooses candidate $\calc_j = \hat{a}_{\ceil*{n/2}}$ located at $y_j$.
	
	Assume without loss of generality that $y_{opt} > y_j$. The social cost of median is:
	\begin{eqnarray*}		
		\SC(M,\bara)=\SC(\calc_j,\bara) &=&\displaystyle \sum_{i: y(\hat{a}_i)\leq y_j}|x_i-y_j| + \sum_{k:y(\hat{a}_k)>y_j}|x_i-y_j|
	\end{eqnarray*}		 
	Denote the first term as $\alpha$, and the second as $\beta$. 
	
	The social cost to the optimal candidate is:
	\begin{eqnarray*}		
		\SC(\OPT,\bx) &=& \displaystyle \sum_{i: y(\hat{a}_i)\leq y_j}|x_i-y_{opt}| + \sum_{k:y(\hat{a}_k)>y_j}|x_i-y_{opt}|
	\end{eqnarray*}			
	Denote the first term by $\gamma$, and the second by $\delta$. 
	
	It is easy to see that $\alpha \leq \gamma$ since for any agent $i$ in these sums (agent whose favorite candidate is not right of $y_j$): $|x_i-y_j| \leq |x_i-y_{opt}|$, and this obviously holds when taking the sum. 
	
	We now show that $\beta \leq \alpha+\gamma+\delta$, due to the following inequalities (justifications for the transitions appear below):
	\begin{eqnarray*}		
		\beta &=&  \sum_{k:\hat{a}_k>y_j}|x_k- y_j| \\
		&\leq& \sum_{k:\hat{a}_k>y_j}|y_j-y_{opt}|  + \sum_{k:\hat{a}_k>y_j}|y_{opt}-x_k| \\
		&=& \sum_{k:\hat{a}_k>y_j} |y_j-y_{opt}|+\delta \\
		&\leq& \sum_{i: y(\hat{a}_i)\leq y_j} |y_j-y_{opt}|+\delta \\
		&\leq& \sum_{i: y(\hat{a}_i)\leq y_j} \left(|y_j-x_i|+|x_i-y_{opt}| \right)+\delta = \alpha+\gamma+\delta
	\end{eqnarray*}	
	The inequalities in the second and fifth lines hold due to the triangle inequality, and the inequality in the fourth line holds because we are summing over a greater or equal amount of non-negative numbers (since $|i:y(\hat{a}_i)\leq y_j| \geq |k:y(\hat{a}_k)> y_j|$).

	Putting this all together: $$\frac{\SC(\calc_j,\bx)}{\SC(\OPT,\bx)} = \frac{\alpha+\beta}{\gamma+\delta} \leq \frac{\gamma+\alpha+\gamma+\delta}{\gamma+\delta} \leq \frac{3\gamma+\delta}{\gamma+\delta}\leq \frac{3\gamma+3\delta}{\gamma+\delta}=3$$
\end{proof}